\documentclass[sigconf,nonacm]{acmart}

\renewcommand\footnotetextcopyrightpermission[1]

\usepackage{hyperref}       % hyperlinks
\usepackage{url}            % simple URL typesetting
\usepackage{booktabs}       % professional-quality tables
\usepackage{amsfonts}       % blackboard math symbols
\usepackage{multirow}       % for multirow in tables
\usepackage{graphicx}
\usepackage{amsfonts}       % blackboard math symbols
\usepackage{nicefrac}       % compact symbols for 1/2, etc.
\usepackage{microtype}      % microtypography
\usepackage{xcolor}         % colors

\usepackage{amsmath}
\usepackage{adjustbox}

\usepackage{enumitem}
\usepackage{algorithm}
\usepackage{algorithmic}
\usepackage{tablefootnote}
\usepackage{subcaption}
\usepackage{graphicx}
\usepackage{makecell}
\usepackage{setspace}
\usepackage[font=small]{caption}

\newcommand{\model}{{{\tt FeDecider}}}
\AtBeginDocument{%
  }

\setcopyright{acmlicensed}
\copyrightyear{2026}
\acmYear{2026}
\setcopyright{cc}
\setcctype{by}
\acmConference[WWW '26]{Proceedings of the ACM Web Conference 2026}{April 13--17, 2026}{Dubai, United Arab Emirates}
\acmBooktitle{Proceedings of the ACM Web Conference 2026 (WWW '26), April 13--17, 2026, Dubai, United Arab Emirates}
\acmPrice{}
\acmDOI{10.1145/3774904.3792294}
\acmISBN{979-8-4007-2307-0/2026/04}

\begin{document}

\title{FeDecider: An LLM-Based Framework for Federated Cross-Domain Recommendation}

\author{Xinrui He}
\affiliation{%
  \institution{University of Illinois Urbana-Champaign}
  \city{Champaign}
  \state{IL}
  \country{USA}
}
\email{xhe33@illinois.edu}

% \affiliation{%

\author{Ting-Wei Li}
\affiliation{%
  \institution{University of Illinois Urbana-Champaign}
  \city{Champaign}
  \state{IL}
  \country{USA}
}
\email{twli@illinois.edu}

\author{Tianxin Wei}
\affiliation{%
  \institution{University of Illinois Urbana-Champaign}
  \city{Champaign}
  \state{IL}
  \country{USA}
}
\email{twei10@illinois.edu}

\author{Xuying Ning}
\affiliation{%
  \institution{University of Illinois Urbana-Champaign}
  \city{Champaign}
  \state{IL}
  \country{USA}
}
\email{xuying2@illinois.edu}

\author{Xinyu He}
\affiliation{%
  \institution{University of Illinois Urbana-Champaign}
  \city{Champaign}
  \state{IL}
  \country{USA}
}
\email{xhe34@illinois.edu}

\author{Wenxuan Bao}
\affiliation{%
  \institution{University of Illinois Urbana-Champaign}
  \city{Champaign}
  \state{IL}
  \country{USA}
}
\email{wbao4@illinois.edu}

\author{Hanghang Tong}
\affiliation{%
  \institution{University of Illinois Urbana-Champaign}
  \city{Champaign}
  \state{IL}
  \country{USA}
}
\email{htong@illinois.edu}

\author{Jingrui He}
\affiliation{%
  \institution{University of Illinois Urbana-Champaign}
  \city{Champaign}
  \state{IL}
  \country{USA}
}
\email{jingrui@illinois.edu}

\renewcommand{\shortauthors}{Xinrui He et al.}
%% No italics, no superscripts, not anonymous
%% Use footnote or author note to identify equal contribution and/or contact author info

\begin{abstract}
Federated cross-domain recommendation (Federated CDR) aims to collaboratively learn personalized recommendation models across heterogeneous domains while preserving data privacy. Recently, large language model (LLM)-based recommendation models have demonstrated impressive performance by leveraging LLMs' strong reasoning capabilities and broad knowledge. However, adopting LLM-based recommendation models in Federated CDR scenarios introduces new challenges. First, there exists a risk of overfitting with domain-specific local adapters. The magnitudes of locally optimized parameter updates often vary across domains, causing biased aggregation and overfitting toward domain-specific distributions. Second, unlike traditional recommendation models (e.g., collaborative filtering, bipartite graph-based methods) that learn explicit and comparable user/item representations, LLMs encode knowledge implicitly through autoregressive text generation training. This poses additional challenges for effectively measuring the cross-domain similarities under heterogeneity. To address these challenges, we propose an LLM-based framework for federated cross-domain recommendation, \model. Specifically, \model\ tackles the challenge of scale-specific noise by disentangling each client’s low-rank updates and sharing only their directional components. To handle the need for flexible and effective integration, each client further learns personalized weights that achieve the data-aware integration of updates from other domains. Extensive experiments across diverse datasets validate the effectiveness of our proposed \model. The implementation is available at \url{https://github.com/Xinrui17/FeDecider}.
\end{abstract}




\maketitle

\vspace{-10pt}
\section{Introduction}
\label{sec:introduction}

In real-world recommendation scenarios \cite{10.1145/3460231.3474268,10.1145/3583780.3615039}, privacy regulations on customer data, including full shopping history and personal information and institutional boundaries \cite{illman2019california} often hinder centralized data collection across platforms, leaving valuable cross-domain information fragmented. Federated cross-domain recommendation, which integrates federated learning \cite{mcmahan2017communication, 10.1145/3589334.3645626} into CDR, provides a promising alternative by training models collaboratively across decentralized clients without exposing raw user data. Many recent studies have explored this direction in the context of traditional recommendation models \cite{chen2022differential,chen2023win}. For example, FedCT \cite{liu2021fedct} adopts a variational inference framework to maximize mutual information across domains; FedGCDR \cite{yang2024federated} enhances privacy and reduces negative transfer through feature mapping and graph neural networks. However, these approaches primarily focus on rigid id-based user or item representation mapping under federated constraints and are typically built on simple collaborative filtering architectures \cite{ricci2010introduction, he2017neural,xue2017deep,wang2019neural,9338366}, which struggle to leverage rich textual information and generalize to new users or items.

\begin{figure}[t]
  \centering
  \includegraphics[width=0.45\textwidth]{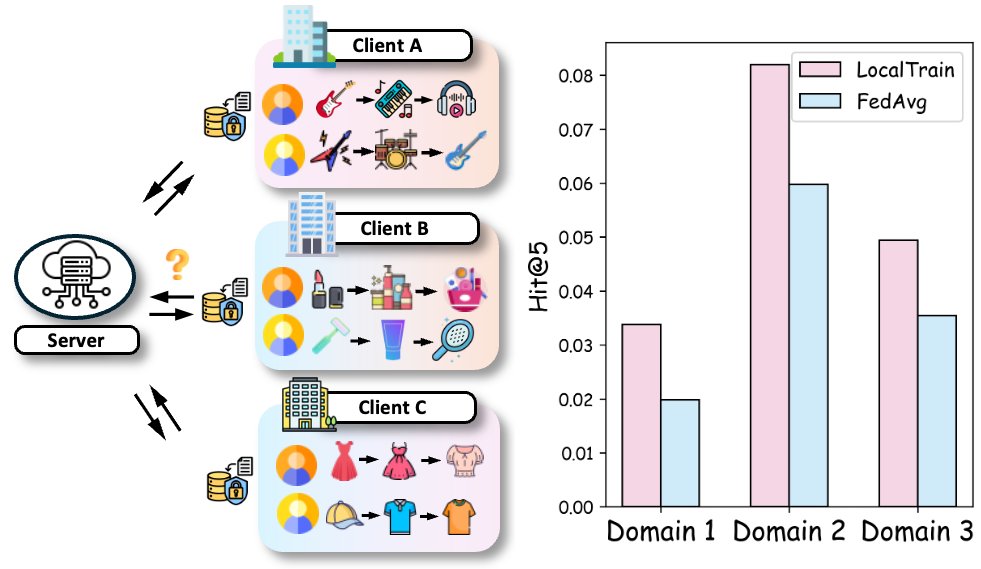}
  \vspace{-4pt}
\caption{\textbf{Left: Illustration of the federated cross-domain recommendation. Right: Comparison of performance (Hit@5) between local training and FedAvg aggregation across 3 domains in Goodreads.}
}
\vspace{-10pt}
  \label{fig:domain_hetero}
\end{figure}

Recently, LLM for recommendation has achieved great success by leveraging its strong reasoning ability in understanding user intent and its broad knowledge for supporting more contextually personalized outputs \cite{li2023text, rajput2023recommender, wang2024learnable, 10.1145/3705328.3748077}. 
As illustrated in Figure~\ref{fig:domain_hetero} (left), different clients represent distinct domains (e.g., Instruments, Beauty, Clothing), where sequential dependencies among items may exhibit similar contextual patterns. For example, the transition from core products to related accessories in one domain can benefit models in another by providing analogous contextual cues.
However, as shown in Figure~\ref{fig:domain_hetero} (right), directly adopting LLM-based recommendation models in federated cross-domain scenario via a standard federated learning method (e.g., FedAvg \cite{mcmahan2017communication}) which averages the uploaded parameters of all clients results in significant performance degradation in all domains compared to local training in the GoodReads dataset \cite{wan2018item}.

This highlights two challenges of transferring and integrating useful knowledge in LLM-based federated recommendation under domain heterogeneity: \textbf{(1) Risk of overfitting with domain-specific local adapters.} Federated LLM-based recommender typically exchanges only partial model parameters, such as LoRA adapters, instead of full-model updates. While efficient, the magnitude (scale) of these locally optimized parameter updates vary significantly across domains, causing biased aggregation and overfitting to local data distributions.
\textbf{(2) Difficulty in cross-domain similarity measurement.} Backbone LLMs model contextual dependencies rather than explicit user-item relations through autoregressive text generation training, unlike traditional recommenders that learn explicit user/item embeddings, comparable representations. This implicit representation hinders direct comparison and aggregation across domains, thereby necessitating data-aware integration to enable effective cross-domain coordination. \textbf{Figure~\ref{fig:challenge}} empirically shows these two challenges, where both aggregation strategies suffer consistent performance degradation across domains due to scale-specific noise and ineffective similarity measurement.

To address these challenges, we propose \model, an LLM-based framework for federated cross-domain recommendation. As shown in Figure~\ref{fig_overview}, each client (domain) first represents its update as a Low-Rank Adaptation (LoRA) module \cite{hu2022lora} and uploads it to the server. Then, to address the first challenge, the \model\ server \textit{extracts only the directional components of the LoRA updates}, which are later broadcast to all clients. This design removes scale-wide noise while preserving the domain-specific pattern.
After each client receives the lightweight LoRA modules from other clients, \model\ clients perform \textit{data-aware integration}. Specifically, clients \textit{learn a set of personalized weights to combine the received directional components with its own}. This mechanism addresses the second challenge by allowing clients to selectively incorporate external knowledge and automatically adapting to their local data distributions. By enabling client-specific integration, \model\ supports flexible knowledge fusion even when semantic representations vary across domains.
The final personalized model is constructed as a weighted combination of all directional components, enabling effective and flexible knowledge sharing under domain heterogeneity.

The major contributions of this paper are summarized as follows:

\begin{figure}[t]
  \centering
  \includegraphics[width=0.40\textwidth]{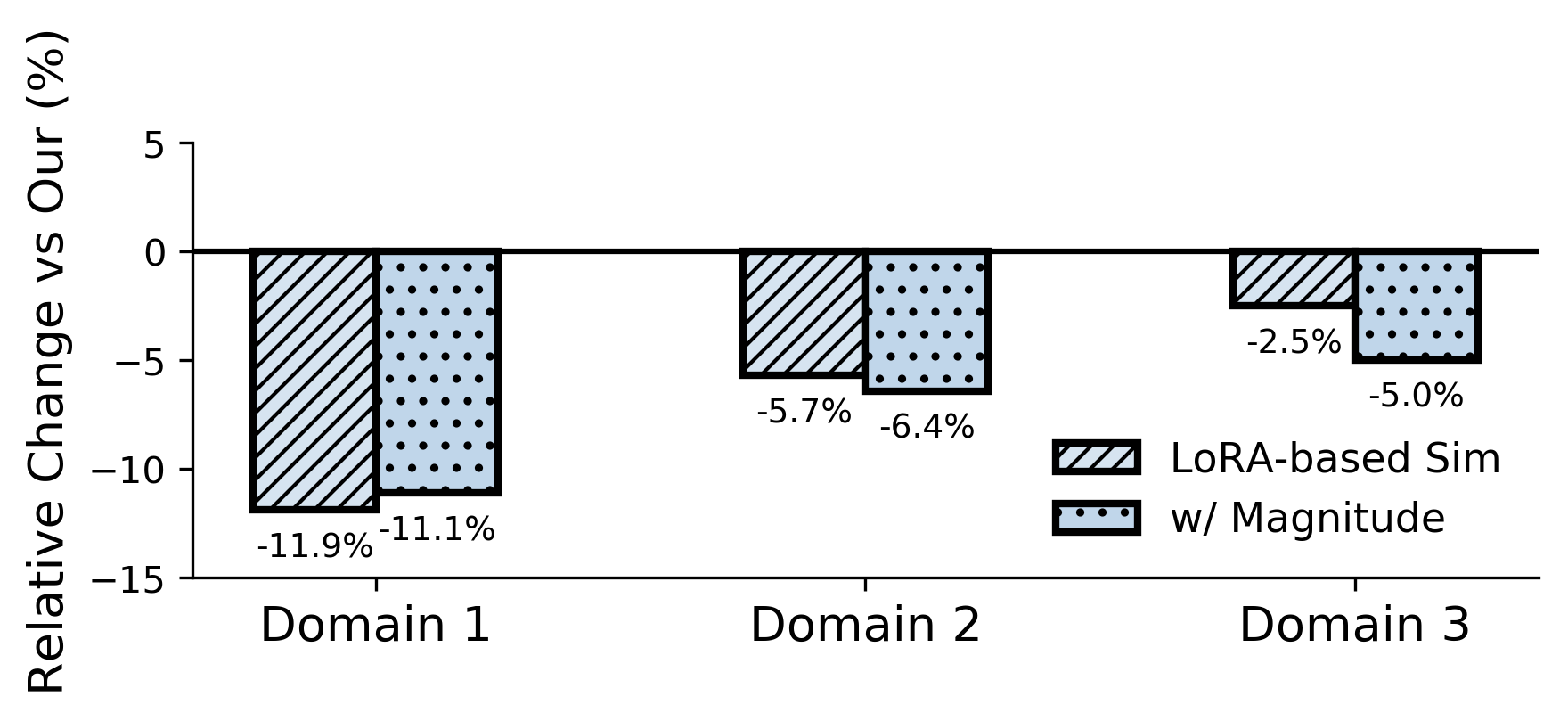}
  \vspace{-5pt}
\caption{\textbf{Performance degradation (NDCG@5) of two baseline aggregation strategies across domains in GoodReads dataset, illustrating the challenges of scale-specific noise and ineffective cross-domain similarity measurement with LoRA. The aggregation weights of the baseline LoRA-based Sim are computed using the layer-wise cosine similarity of the LoRA module.}
}
\vspace{-15pt}
  \label{fig:challenge}
\end{figure}

% \vspace{-0.3cm}
\begin{itemize}[leftmargin=10pt]
    \item We introduce \model, a novel LLM-based federated cross-domain recommendation model that addresses the unique challenges when adopting LLM-based recommendation models in federated CDR, including (i) the risk of overfitting with domain-specific local adapters and (ii) the intrinsic difficulty of cross-domain similarity computation.

    \item \model \textit{only communicates directional components of LoRA updates} across clients to efficiently and effectively transfer the knowledge embedded in the LLM-based recommendation models. Then, each client \textit{performs data-aware integration with locally learnable personalized weights}, allowing fine-grained integration.
  
    \item Experiments on multiple cross-domain recommendation datasets confirm the strong performance and robustness of \model.
\end{itemize}

\vspace{-5pt}
\section{Preliminaries}
\label{sec:preliminary}
\vspace{-5pt}
\subsection{Cross-Domain Recommendation}
In practical settings, cross-domain recommendation typically involves a small number of closely related domains (e.g., books, movies, music), each with its own item space. Despite being disjoint, these domains often share latent semantics such as user preferences and content themes that can benefit knowledge transfer across domains. We consider a federated cross-domain recommendation task with \( K \) clients, where each client \( \mathcal{D}_i \) represents a distinct domain for \( i = 1, 2, \ldots, K \). Each domain maintains its own user set \( \mathcal{U}_i \), item pool \( \mathcal{V}_i = \{v_{i,1}, v_{i,2}, \ldots\} \), and interaction dataset \( \mathcal{R}_i \subseteq \mathcal{U}_i \times \mathcal{V}_i \). Different from general federated learning, which learns a single model across all clients (domains), federated cross-domain recommendation addresses domain heterogeneity by learning a personalized model \(W_i\) with \( W_i = \Theta + \Delta W_i \) for each client \( \mathcal{D}_i \). To accommodate client-side resource constraints and reduce communication overhead, we adopt a shared LLM backbone with frozen parameters \( \Theta \), and fine-tune lightweight, client-specific LoRA modules \( \Delta W_i \) on local data. The training objective is to minimize the negative log-likelihood of next-item prediction:
\begin{equation*}
\label{eq:1}
    \min_{\{ \Delta W_i \}_{i=1}^K} \ \frac{1}{K} \sum_{i=1}^K \mathbb{E}_{(u, s_u, v^+) \sim \mathcal{R}_i} \left[ -\log P_{W_i}(v^+ \mid s_u) \right],
\end{equation*}

where \( s_u \) is the interaction history of user \( u \in \mathcal{U}_i \), and \( v^+ \in \mathcal{V}_i \) is the ground-truth next item.
\begin{figure*}[t]
  \centering

  \includegraphics[width=0.95\textwidth]{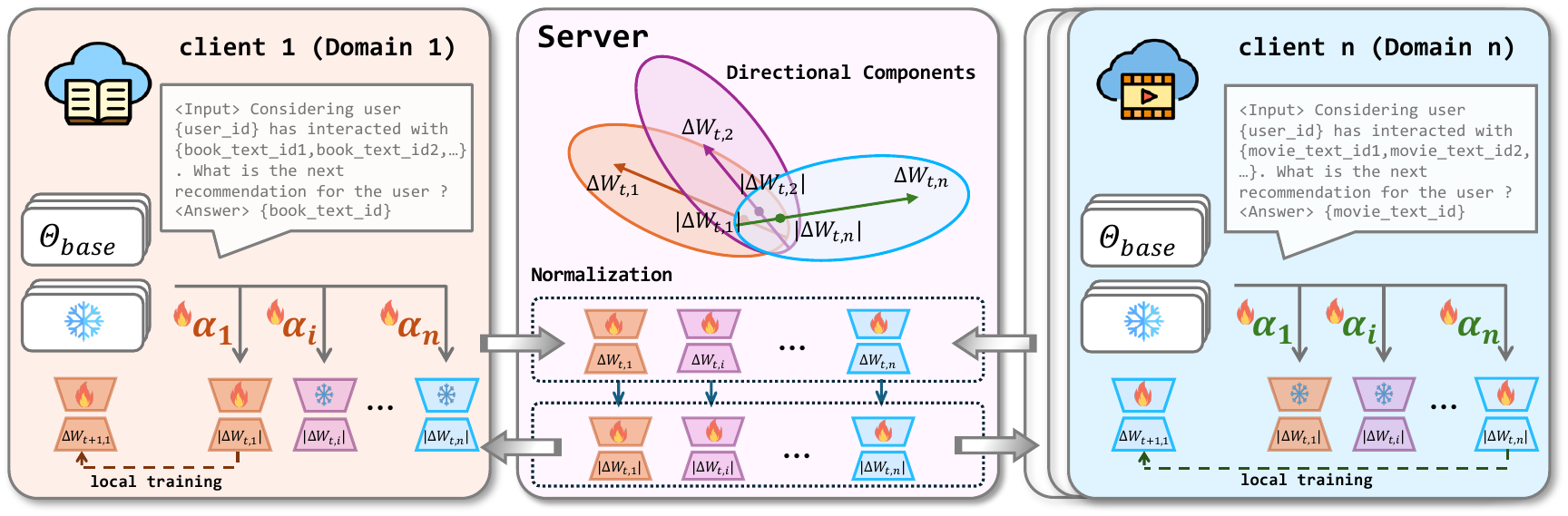}
  \vspace{-10pt}
  \caption{Overview of \model. \textbf{Server (\textit{middle}):} At each federated round, the server collects (\textit{grey arrow}) LoRA updates from all clients, extracts their directional components, and then broadcasts them to all clients (\textit{grey arrow}). \textbf{Clients (\textit{left and right}):} Each client (e.g., Domain 1 for books and Domain n for movies) keeps the base model frozen and integrates the received directional components using a set of personalized weights \(\alpha\), and only updates its own LoRA and \(\alpha\) during local training.}
\vspace{-10pt}
  \label{fig_overview}
\end{figure*}
\vspace{-5pt}
\subsection{Low-Rank Adaptation in LLM Fine-tuning}

Low-Rank Adaptation (LoRA)~\cite{hu2022lora} is a widely applied parameter-efficient tuning method for LLM finetuning. It freezes the original model weights and inserts trainable low-rank matrices into linear layers. For a weight matrix \( W \in \mathbb{R}^{d \times d} \) in LLM, LoRA introduces LLM update as \( \Delta W = B A \), where \( A \in \mathbb{R}^{r \times d} \), \( B \in \mathbb{R}^{d \times r} \) are trainable parameters, and \( r \ll d \). Given an input vector \( x \in \mathbb{R}^d \), the output in inference stage becomes \(\tilde{W} x = (W + B A)x\). Note that only \( A \) and \( B \) are updated during finetuning, which reduces computation costs compared with full finetuning while maintaining downstream performance. In federated LLM fine-tuning, LoRA further reduces communication costs by transmitting only the compact \( A \) and \( B \) matrices rather than full model gradients or weights. Given its wide adoption and efficiency, our framework is built on LoRA-based fine-tuning which significantly reduces communication overhead by updating only a small number of trainable parameters.

\vspace{-5pt}
\section{Method}
In this section, we present the proposed \model~for LLM-based federated cross-domain recommendation. We first introduce the transfer of directional components and then the data-aware integration based on client-specific learnable weights.

\label{sec:method}

\vspace{-5pt}
\subsection{FeDecider}

\paragraph{Communication via Directional Components.} A key challenge in LLM-based federated CDR lies in how clients properly exchange local LoRA updates under domain heterogeneity. Each client's local parameter reflects its unique data distribution and naively utilizing such updates often leads to biased aggregation or even overfitting to domains with dominating parameter magnitude. 
The coupling between the \emph{magnitude} and \emph{direction component} of parameter changes restricts each client's ability to selectively incorporate useful external knowledge while adapting to its own domain characteristics.
To mitigate this issue, we draw inspiration from the decomposition principle introduced in DoRA~\cite{liu2024dora}, which separates a LoRA update into directional and magnitude components. Specifically, we represent each client's low-rank update as \( \Delta W_j = a_j B_j A_j \), where \( A_j \in \mathbb{R}^{r \times d} \) and \( B_j \in \mathbb{R}^{d \times r} \) encode the update direction, and \( a_j \in \mathbb{R} \) is a scalar representing the update magnitude.
To extract the directional components without disrupting the separability of LoRA parameters or increasing communication cost, we adopt a server-side normalization strategy that jointly rescales both matrices:
\[
\tilde{A_j} \leftarrow s \cdot A_j, \quad \tilde{B_j} \leftarrow s \cdot B_j, \quad \text{where} \quad s := \left( \frac{1}{\|B_j A_j\|_F} \right)^{1/2}.
\]
This ensures that \(
\|\tilde{B_j} \tilde{A_j}\|_F = 1
\) and each \( \tilde{B_j} \tilde{A_j} \) represents the client’s normalized update direction. In Appendix~A.1--A.2, we provide an analytical perspective on the effect of separating update scale and direction under heterogeneity.

\paragraph{Data-aware Integration via Learnable Weights} Another challenge arises in measuring cross-domain similarity for effective integration. Specifically, backbone LLMs in federated CDR primarily capture contextual dependencies for recommendation generation rather than explicit user–item relations as in traditional recommenders that learn comparable embeddings across domains. This implicit representation in LLM hinders the assessment of alignment among client-specific updates, often resulting in suboptimal or inconsistent aggregation.
To resolve this, we leverage the previously introduced local directional components to enable client-side adaptive integration. In particular, each client learns personalized weights to selectively incorporate external directions into its own LoRA modules, ensuring that cross-domain information is aligned to its local distribution. Formally, for client \( i \), we learn personalized weights \( \alpha_{ij} \) for directional components of client \( j \in \{1, \ldots, N\} \), forming the update with the following form:
\begin{equation*}
  \Delta W_i = \sum_{j=1}^N \alpha_{ij} \cdot \tilde{B_j} \tilde{A_j}.  
\end{equation*}
Here, each \( \tilde{A_j}, \tilde{B_j} \) is the normalized directional component from client \( j \). During local training, client \( i \) further updates (1) its own directional matrices \( \tilde{A_i}, \tilde{B_i} \), and (2) all personalized weights \( \{ \alpha_{ij} \}_{j=1}^N \). Directional components from other clients (\( j \ne i \)) are kept fixed to preserve shared knowledge. See Appendix \ref{app:personalized_weight} for an intuitive explanation of why gradient updates on $\alpha_{ij}$ tend to decrease weights on locally harmful directions and increase weights on beneficial ones. Finally, the local training objective is:

\noindent
\resizebox{0.88\linewidth}{!}{
\centering
\begin{minipage}{\linewidth}
\begin{equation*}
\label{eq:2}
\begin{aligned}
    \min_{\alpha_{ij}, \tilde{A_i}, \tilde{B_i}} \; F_i &= \mathbb{E}_{(x, y) \sim \mathcal{D}_i} \left[ \mathcal{L} \left(\left(x; W_0 + \Delta W_i\right), y \right) \right]
\text{ ,with } \Delta W_i = \sum_{j=1}^N \alpha_{ij} \cdot \tilde{B_j} \tilde{A_j},
\end{aligned}
\end{equation*}
\end{minipage}
}

\noindent
where \( x \) denotes the generative recommendation input constructed from user interaction history and \( y \) is the target item to be predicted. This optimization enforces that each client personalizes its own update while learning how to incorporate shared knowledge from other clients through learnable weights \( \alpha_{ij} \). We provided a formal analysis on how personalized weights can down-weight the harmful directions in local integration.

\paragraph{Overall Training Process}
The federated training process is summarized in Figure~\ref{fig_overview}. Before communication begins, each client independently trains its own LoRA matrices \( A_i \) and \( B_i \) on local data. At the end of the first local training phase, clients upload their raw LoRA updates to the server. The server then performs normalization to extract directional components \( \tilde{A_j}, \tilde{B_j} \). This step also prevents other clients from recovering the exact original updates, and additional privacy discussion is provided in \ref{privacy_analysis}. The server then broadcasts these directional components to all clients. Upon receiving these directional components, each client constructs a fine-grained integration by learning its personalized set of aggregation weights \( \{ \alpha_{ij} \} \) and updates its own LoRA.

In subsequent communication rounds, each client uploads its updated LoRA to the server, while the server re-applies normalization and redistributes the new directional components. This process is repeated iteratively. After the final communication round, each client’s personalized model is defined as a weighted combination of all directional components using the learned aggregation weights, without any further local fine-tuning. A detailed discussion on communication cost is provided in Section \ref{communication_cost}.

\begin{table*}[ht]

\setlength{\abovecaptionskip}{0.05cm}
\setlength{\belowcaptionskip}{0.2cm}
\caption{Overall performance comparison between the baselines on GoodReads Crime \& Comics \& Children. The best results are highlighted in boldface, and the second-best results are underlined.} \label{tab:main-goodreads}
\setlength{\tabcolsep}{3mm}{
\resizebox{\textwidth}{!}{
\begin{tabular}{c|cccc|cccc|cccc|c}
\toprule
% \hline
% \textbf{Dataset}&\multicolumn{12}{c}{\textbf{Googreads}}\\ \midrule
 \textbf{Domain} &
 %  \multicolumn{4}{c|}{} &
 %  \multicolumn{4}{c|}{} &
 %  \multicolumn{4}{c}{} \\
 % &
  \multicolumn{4}{c|}{\multirow{-1}{*}{\textbf{GoodReads- Crime}}} &
  \multicolumn{4}{c|}{\multirow{-1}{*}{\textbf{GoodReads- Comics}}} &
  \multicolumn{4}{c|}{\multirow{-1}{*}{\textbf{GoodReads- Children}}} & \multicolumn{1}{c}{\textbf{Avg}}\\
\multicolumn{1}{c|}{\textbf{Metric}} &
  \multicolumn{1}{c}{\textbf{H@5}} &
  \multicolumn{1}{c}{\textbf{H@10}} &
  \multicolumn{1}{c}{\textbf{N@5}} &
  \multicolumn{1}{c|}{\textbf{N@10}} &
  \multicolumn{1}{c}{\textbf{H@5}} &
  \multicolumn{1}{c}{\textbf{H@10}} &
  \multicolumn{1}{c}{\textbf{N@5}} &
  \multicolumn{1}{c|}{\textbf{N@10}} &
  \multicolumn{1}{c}{\textbf{H@5}} &
  \multicolumn{1}{c}{\textbf{H@10}} &
  \multicolumn{1}{c}{\textbf{N@5}} &
  \multicolumn{1}{c}{\textbf{N@10}} &
  \multicolumn{1}{c}{\textbf{H@5}}\\   \midrule\midrule
Local train &
  0.0339 &
  0.0426 &
  0.0232 &
  0.0260 &
  0.0820 &
  \underline{0.1071} &
  0.0625 &
  0.0706 &
  0.0494 &
  \underline{0.0624} &
 0.0376 &
  0.0418 & 0.0551
  \\

FedAvg &
  0.0198 &
 0.0252&
 0.0154&
0.0171&
 0.0598&
0.0790&
 0.0458&
0.0520 &
0.0354&
0.0456&
0.0283 &
0.0315& 0.0383
  \\
PFedAvg &
  0.0328 &
  0.0420 &
  0.0223 &
  0.0236 &
  0.0804 &
  0.1006 &
  0.0620 &
  0.0684 &
  0.0486 &
  0.0614 &
  0.0367 &
  0.0409 & 0.0539
  \\

FedProx &
  0.0340 &
  0.0412 &
  0.0233 &
  0.0257 &
  0.0838 &
  0.1056 &
  \underline{0.0661} &
  \underline{0.0738} &
  \underline{0.0500} &
  0.0620 &
  0.0381 &
  0.0420 & \underline{0.0559}
  \\
Ditto &
  \textbf{0.0358} &
  \underline{0.0452} &
  \underline{0.0241} &
  \underline{0.0271} &
  0.0836 &
  0.1036 &
  0.0651 &
  0.0716 &
  0.0482 &
  0.0592 &
  0.0368 &
  0.0403 & \underline{0.0559}
  \\
FFA-LoRA&
  0.0304 &
  0.037 &
  0.0200 &
  0.0222 &
  0.0676 &
  0.0888 &
  0.0517 &
  0.0586 &
  0.0374 &
  0.0468 &
  0.0297 &
  0.0328 & 0.0451
  \\
Rolora &
  0.0322 &
  0.0408 &
  0.0217 &
  0.0245 &
  0.0820 &
  0.1032 &
  0.0635 &
  0.0703 &
  0.0494 & 
  \textbf{0.0640} &
   \underline{0.0378} &
  \underline{0.0425} & 0.0545
  \\
FellRec&
  0.0316 &
  0.0406 &
  0.0211 &
  0.0240 & 
  \underline{0.0840} &
  0.1056 &
  0.0637 &
  0.0706 &
  0.0474 &
  0.0618 &
  0.0359 &
  0.0406 & 0.0543
  \\
FDLoRA &
  0.0322 &
  0.0396 &
  0.0215 &
  0.0239 &
  0.0818 &
  0.1054 &
  0.0620 &
  0.0697 &
  \underline{0.0500} &
  0.0620 &
  0.0373 &
  0.0412 & 0.0547
  
 \\ \midrule
\model &
\underline{0.0352} &
\textbf{0.0478} &
\textbf{0.0252} &
\textbf{0.0292} &
\textbf{0.0862} &
  \textbf{0.1078} &
  \textbf{0.0701} &
  \textbf{0.0771} &
  \textbf{0.0502} &
  0.0620 &
\textbf{0.0401} &
\textbf{0.0439} &
\textbf{0.0572}\\  \bottomrule
\end{tabular}
}}
\end{table*}

\begin{table*}[t]

\setlength{\abovecaptionskip}{0.05cm}
\setlength{\belowcaptionskip}{0.2cm}
\caption{Overall performance comparison between the baselines on Amazon Beauty \& Clothing. The best results are highlighted in boldface, and the second-best results are underlined.} \label{tab:main-amazon}
\setlength{\tabcolsep}{3mm}{
\begin{adjustbox}{width=0.8\textwidth,center} 
\begin{tabular}{c|cccc|cccc|c}
\toprule
% \hline
% \textbf{Dataset}&\multicolumn{12}{c}{\textbf{Googreads}}\\ \midrule
 \textbf{Domain} &
 %  \multicolumn{4}{c|}{} &
 %  \multicolumn{4}{c|}{} &
 %  \multicolumn{4}{c}{} \\
 % &
  \multicolumn{4}{c|}{\multirow{-1}{*}{\textbf{Amazon- Beauty}}} &
  \multicolumn{4}{c|}{\multirow{-1}{*}{\textbf{Amazon- Clothing}}} &
  \multicolumn{1}{c}{\textbf{Avg}}\\
\multicolumn{1}{c|}{\textbf{Metric}} &
  \multicolumn{1}{c}{\textbf{H@5}} &
  \multicolumn{1}{c}{\textbf{H@10}} &
  \multicolumn{1}{c}{\textbf{N@5}} &
  \multicolumn{1}{c|}{\textbf{N@10}} &
  \multicolumn{1}{c}{\textbf{H@5}} &
  \multicolumn{1}{c}{\textbf{H@10}} &
  \multicolumn{1}{c}{\textbf{N@5}} &
  \multicolumn{1}{c|}{\textbf{N@10}} &
  \multicolumn{1}{c}{\textbf{H@5}}\\   \midrule\midrule
Local train &
  0.0172 &
  0.0220 &
  0.0124 &

  0.0139 &
  0.0080 &
  0.0107 &
  0.0057 &
  0.0062 & 0.0126
    \\
FedAvg &
 0.0168 &
 0.0220&
  0.0130&
 \underline{0.0146}&
0.0068 &
 0.0100 &
  0.0054&
 0.0064& 0.0118
  \\
PFedAvg &
  0.0173 &
  0.0210 &
  \underline{0.0132} &
  0.0144 &
  0.0078 &
  0.0105 &
  0.0062 &
  0.0071 & 0.0126
 
  \\
FedProx &
  0.0160 &
  0.0200 &
  0.0121 &
  0.0134 &
  0.0075 &
  0.0100 &
  0.0060 &
  0.0068 & 0.0118

  \\
Ditto &
  0.0158 &
  0.0210 &
  0.0122 &
  0.0140 &
  0.0088 &
  0.0113 &
  0.0067 &
  0.0075 & 0.0123
  
  \\
FFA-LoRA &
  0.0140 &
  0.0163 &
  0.0109 &
  0.0116 &
  0.0073 &
  \underline{0.0120} &
  0.0052 &
  0.0067 & 0.0107
  
  \\
Rolora &
  0.0165 &
  \underline{0.0218} &
  0.0124 &
  0.0141 &
  0.0083 &
  0.0103 &
  0.0065 &
  0.0072 & 0.0124
  
  \\
FellRec &
  \underline{0.0173} &
  0.0210 &
  0.0128 &
  0.0140 &
  0.0085 &
  0.0105 &
  \underline{0.0068} &
  0.0075 & \underline{0.0129}
  
  \\
FDLoRA &
  0.0163 &
  0.0205 &
  0.0123 &
  0.0143 &
  \underline{0.0090} &
  0.0115 &
  \textbf{0.0069} &
  \underline{0.0077} & 0.0127

 \\ \midrule
\model &
\textbf{0.0190} &
\textbf{0.0247} &
\textbf{0.0145} &
\textbf{0.0164} &
\textbf{0.0097} &
  \textbf{0.0161} &
  0.0061 &
  \textbf{0.0084} &
  
\textbf{0.0144}\\  \bottomrule
\end{tabular}
\end{adjustbox}}
\vspace{-5pt}
\end{table*}

\section{Results}
\label{sec:results}

In this section, we present experimental results on several real-world cross-domain datasets to evaluate the performance of \model. Our experiments intend to answer the following research questions:

\begin{itemize}[leftmargin=*]
    \item \textbf{RQ1:} How does \model\ perform compared to existing LLM-based models for federated cross-domain recommendation tasks?
    \item \textbf{RQ2:} To what extent does each individual component in \model\ contribute to the overall performance?
    \item \textbf{RQ3:} How effective are the personalization weights learned by \model\ in capturing domain similarities?
    \item \textbf{RQ4:} How robust is \model\ across user groups with different activity levels and item groups with varying popularity?
\end{itemize}

\vspace{-10pt}
\subsection{Experiment Setup}

\paragraph{Datasets}
We conduct experiments on Goodreads \cite{wan2018item} and the Amazon review dataset \cite{he2016ups}. Specifically, we select three domains, \textit{Crime}, \textit{Comics}, and \textit{Children}, from Goodreads and two pairs of domains from Amazon, which are \textit{Beauty} $\&$ \textit{Clothing} and \textit{Electronics} $\&$ \textit{Phones}, to evaluate the performance of \model~in federated cross-domain recommendation scenarios. The data statistics and preprocessing procedure are introduced in Appendix \ref{appendix:data}, and results of \textit{Electronics} $\&$ \textit{Phones} are provided in Appendix \ref{appendix:add_results}.

\paragraph{Baseline Methods} We compare our approach with seven representative personalized federated learning baselines, including both classic algorithms and recent LoRA-specific federated fine-tuning methods. We applied each baseline to LLM-based cross-domain recommendation tasks. For classic methods, we instantiate their local-update and aggregation rules within our parameter-efficient fine-tuning pipeline by applying them to the LoRA adapter parameters. Below, we briefly summarize each baseline. \textbf{FedAvg} \cite{mcmahan2017communication} is the basic algorithm for federated learning, which performs global aggregation of locally trained models by averaging their parameters.
\textbf{PFedAvg} is a personalized variant of FedAvg that applies local adaptation after the final global aggregation to capture user-specific preferences. \textbf{FedProx} \cite{li2020federated} extends FedAvg with a proximal term to constrain local updates and mitigate client drift caused by data heterogeneity. \textbf{Ditto} \cite{li2021ditto} maintains both a global and a personalized model on each client using bi-level optimization to improve personalization. \textbf{FFA-LoRA}\cite{sun2024improving} freezes LoRA-A matrices and only fine-tunes LoRA-B to address the instability of LoRA aggregation in FL.
% caused by the mismatch between local client updates and global aggregation
\textbf{Rolora} \cite{chenrobust} improves LoRA-based federated fine-tuning via alternating aggregation for better robustness and efficiency. \textbf{FellRec} \cite{zhao2025federated} introduces dynamic balancing and flexible storage to improve the performance of LLM-based federated recommendation. \textbf{FDLoRA} \cite{qi2024fdlora} employs dual LoRA modules to separate global and personalized learning and uses adaptive fusion to balance performance and efficiency.

\paragraph{Implementation Details}

We adopt IDGenRec~\cite{tan2024idgenrec} as the backbone model. During the pre-training stage, we jointly fine-tune the base recommendation model and the ID generator using a small subset of data randomly sampled from the interactions that are not included in any local client. Specifically, we use 2{,}000 interactions for Amazon and 3{,}000 interactions for Goodreads. For the pre-training, we follow the default hyperparameters provided in the original IDGenRec implementation. Considering that clients may have limited resources to store the full model, we follow SLMRec~\cite{xu2025slmrec} and distill a smaller model with half the parameter size using the same pre-training dataset, which serves as the base model for all clients. Details of the distillation process are provided in Appendix \ref{storage_efficiency}. After pre-training, the ID generator is fixed and shared across all clients without further tuning. For client fine-tuning, LoRA layers are inserted into all attention layers of the backbone model. For a fair comparison, all methods perform one round of federated communication after every five local training epochs, totaling 20 rounds on Goodreads and 25 rounds on Amazon. The local batch size is set to 64, and the learning rate is 1e-3. For our method, the personalization weights are initialized to 2 and optimized with a learning rate of 1e-3. All experiments are conducted on a machine with 4 NVIDIA A100 GPUs.

\vspace{-5pt}
\subsection{Experimental Results}

\subsubsection{Main Results (RQ1)}
As shown in Table~\ref{tab:main-goodreads} and Table~\ref{tab:main-amazon}, \model{} consistently outperforms all baseline methods across most evaluation metrics on both the Goodreads and Amazon   datasets. In particular, our model achieves the highest performance on all NDCG@10 metrics. The detailed analysis yields the following observations:

(1) Compared to existing frameworks, including both general PFL methods and LoRA-specific frameworks, our approach demonstrates superior performance, \textit{confirming the benefit of our data-aware integration and fine-grained control over shared knowledge}. (2) Personalized baselines such as FedProx, Ditto, and even purely local training perform competitively in some domains, reflecting the advantage of preserving local models under significant domain shifts. This highlights the \textit{importance of balancing personalization and knowledge sharing} in federated CDR. (3) FFA-LoRA and RoLoRA, though designed to mitigate interference and instability during LoRA aggregation, show limited performance, potentially due to reduced information exchange across clients. (4) FellRec achieves good results in several domains, supporting \textit{the utility of weighted aggregation under client heterogeneity}. (5) \textit{More fine-grained shared knowledge is needed}. FDLoRA shows promising performance in domains such as Amazon-Clothing and Goodreads-Children by maintaining personalized LoRA modules and employing an adaptive fusion strategy, but it captures shared knowledge at a coarser granularity than \model. (6) We further compare \model~with ID-based federated cross-domain recommendation frameworks that represent users and items solely through IDs. The results presented in Section \ref{sec:id_baselines} show that \textit{the incorporation of textual information through LLM significantly boosts performance}, underscoring the potential of LLM-based approaches in federated CDR tasks.

\begin{table*}[t]
\setlength{\abovecaptionskip}{0.05cm}
\setlength{\belowcaptionskip}{0.2cm}
\caption{Ablation study of \model\ on GoodReads Crime \& Comics \& Children. The best results are highlighted in boldface, and the second-best results are underlined.} \label{tab:ablation-goodreads}
\setlength{\tabcolsep}{3mm}{
\resizebox{\textwidth}{!}{
\begin{tabular}{c|cccc|cccc|cccc|c}
\toprule
% \hline
% \textbf{Dataset}&\multicolumn{12}{c}{\textbf{Googreads}}\\ \midrule
 \textbf{Domain} &
 %  \multicolumn{4}{c|}{} &
 %  \multicolumn{4}{c|}{} &
 %  \multicolumn{4}{c}{} \\
 % &
  \multicolumn{4}{c|}{\multirow{-1}{*}{\textbf{GoodReads- Crime}}} &
  \multicolumn{4}{c|}{\multirow{-1}{*}{\textbf{GoodReads- Comics}}} &
  \multicolumn{4}{c|}{\multirow{-1}{*}{\textbf{GoodReads- Children}}} & \multicolumn{1}{c}{\textbf{Avg}}\\
\multicolumn{1}{c|}{\textbf{Metric}} &
  \multicolumn{1}{c}{\textbf{H@5}} &
  \multicolumn{1}{c}{\textbf{H@10}} &
  \multicolumn{1}{c}{\textbf{N@5}} &
  \multicolumn{1}{c|}{\textbf{N@10}} &
  \multicolumn{1}{c}{\textbf{H@5}} &
  \multicolumn{1}{c}{\textbf{H@10}} &
  \multicolumn{1}{c}{\textbf{N@5}} &
  \multicolumn{1}{c|}{\textbf{N@10}} &
  \multicolumn{1}{c}{\textbf{H@5}} &
  \multicolumn{1}{c}{\textbf{H@10}} &
  \multicolumn{1}{c}{\textbf{N@5}} &
  \multicolumn{1}{c}{\textbf{N@10}} &
  \multicolumn{1}{c}{\textbf{H@5}}\\   \midrule\midrule
PFedAvg (baseline) &
  0.0328 & 
  0.0420 &
  0.0223 &
  0.0236 &
  0.0804 &
  0.1006 &
  0.0620 &
  0.0684 &
  0.0486 &
  0.0614 &
  0.0367 &
  0.0409 & 0.0539
  \\
\textbf{\model\ w/o Decomp} &
0.0338 &
 0.0426&
 0.0224&
 0.0252&
 \underline{0.0836}&
 \textbf{0.1090}&
 0.0656&
 \underline{0.0737}&
 \textbf{0.0510}&
 \textbf{0.0632}&
 \underline{0.0381}&
 0.0420& \underline{0.0561}
  \\
\textbf{\model\ w/o Per} &
 \underline{0.0342} &
 \underline{0.0460}&
 \underline{0.0246}&
 \underline{0.0283}&
 0.0824&
 0.1050&
 \underline{0.0664}&
 \underline{0.0737}&
 0.0476&
 0.0586&
 0.0360&
 0.0396& 0.0547
  \\

  \textbf{\model\ w/o Sep} &
  \underline{0.0342}&
 0.0428&
 0.0227&
 0.0255&
 0.0830&
 0.1066&
 0.0654&
 0.0731&
 0.0480&
 0.0616&
 0.0376&
 \underline{0.0421}& 0.0551
 \\ \midrule
\model &
\textbf{0.0352} &
\textbf{0.0478} &
\textbf{0.0252} &
\textbf{0.0292} &
\textbf{0.0862} &
  \underline{0.1078} &
  \textbf{0.0701} &
  \textbf{0.0771} &
  \underline{0.0502} &
  \underline{0.0620} &
\textbf{0.0401} &
\textbf{0.0439} &
\textbf{0.0572}\\  \bottomrule
\end{tabular}
}}
\label{tab:MainTable}
\vspace{-5pt}
\end{table*}

\subsubsection{Ablation Studies (RQ2)}

To evaluate the impact of key designs in \model, we conduct experiments with three model variants on Goodreads (Table~\ref{tab:ablation-goodreads}). In all settings, for each client \( i \), only the personalization weights \( \{\alpha_{ij}\}_{j=1}^N \) and its own LoRA parameters \( A_i \), \( B_i \) are trainable, while all external parameters from other clients remain fixed. We study the following variants:

\textbf{(1) w/o Decomp:} This variant removes the directional decomposition. The local training objective becomes:

\noindent
\resizebox{0.95\linewidth}{!}{
\begin{minipage}{\linewidth}
\begin{equation*}
\begin{aligned}
    \min_{\alpha_{ij}, A_i, B_i} \; \mathbb{E}_{(x, y) \sim \mathcal{D}_i} \left[  \mathcal{L} \left(\left(x; W_0 + \Delta W_i\right), y \right) \right] 
 \text{ ,with } \Delta W_i & = \sum_{j=1}^N \alpha_{ij} \cdot B_j A_j.
\end{aligned}
\end{equation*}
\end{minipage}
}

This leads to a slight drop in performance, suggesting that separating direction and weight is beneficial for achieving stable and fine-grained control during aggregation. Without directional decomposition, variations in the scale of update directions across clients can introduce noise during the knowledge transferring, which confirms the fisrt challenge we discussed in Section \ref{sec:introduction}.

\textbf{(2) w/o Per:} This setting removes client-specific personalization weights \( \alpha_{ij} \), and aggregates external LoRA directions uniformly:

% \noindent
\resizebox{0.95\linewidth}{!}{
\begin{minipage}{\linewidth}
\centering
\begin{equation*}
\begin{aligned}
    \min_{\tilde{A}_i, \tilde{B}_i} \; \mathbb{E}_{(x,y) \sim \mathcal{D}_i} \left[  \mathcal{L} \left(\left(x; W_0 + \Delta W_i\right), y \right) \right] 
    \text{ ,with } \Delta W_i & = \frac{1}{N} \sum_{j=1}^N \tilde{B}_j \tilde{A}_j.
\end{aligned}
\end{equation*}
\end{minipage}
}

The performance degradation in this setting confirms the importance of personalization in the aggregation process. Uniformly combining directions ignores inter-client differences and fails to capture domain-specific preferences, especially under heterogeneous data distributions.

\textbf{(3) w/o Sep:} This variant removes the aggregation over complete client updates \( A_j B_j \). Instead, it first aggregates LoRA factors \(A_j\) and \(B_j\) separately and then composes the update:

\noindent
\resizebox{0.85\linewidth}{!}{
\begin{minipage}{\linewidth}
\begin{equation*}
\begin{aligned}
    \min_{\alpha_{ij}, A_i, B_i} \; \mathbb{E}_{(x,y) \sim \mathcal{D}_i} \left[ \mathcal{L} \left(\left(x; W_0 + \Delta W_i\right), y \right) \right]
     \text{ ,with }   \Delta W_i & = (\sum_{j=1}^N \alpha_{ij} A_j ) ( \sum_{j=1}^N \alpha_{ij} B_j ). \quad 
\end{aligned}
\end{equation*}
\end{minipage}
}

This results in degraded performance, which may be caused by breaking the structural alignment between \( A_j \) and \( B_j \) from the same client. This separate aggregation introduces semantic inconsistency and noise in the reconstructed updates.

These ablation results collectively demonstrate that each design component plays a critical role in the effectiveness of \model.

\subsubsection{Analysis of Personalization Weights (RQ3)}

\begin{figure}[t]
    \centering
    \begin{subfigure}[b]{0.155\textwidth}
        \centering
        \includegraphics[width=\linewidth]{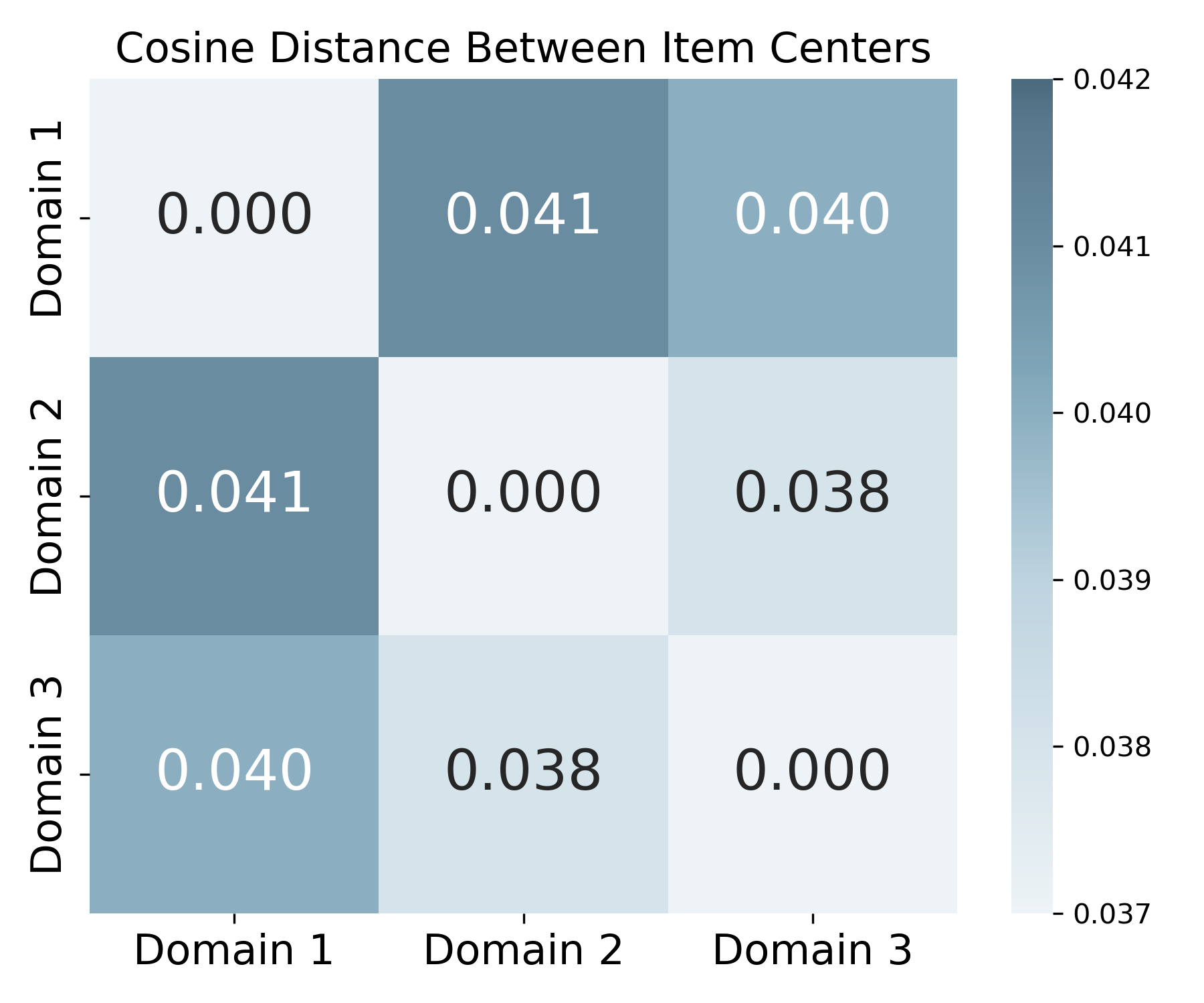}
        % \caption{Cosine Distance Between Item Centers}
       
    \end{subfigure}
    \hfill
    \begin{subfigure}[b]{0.155\textwidth}
        \centering
        \includegraphics[width=\linewidth]{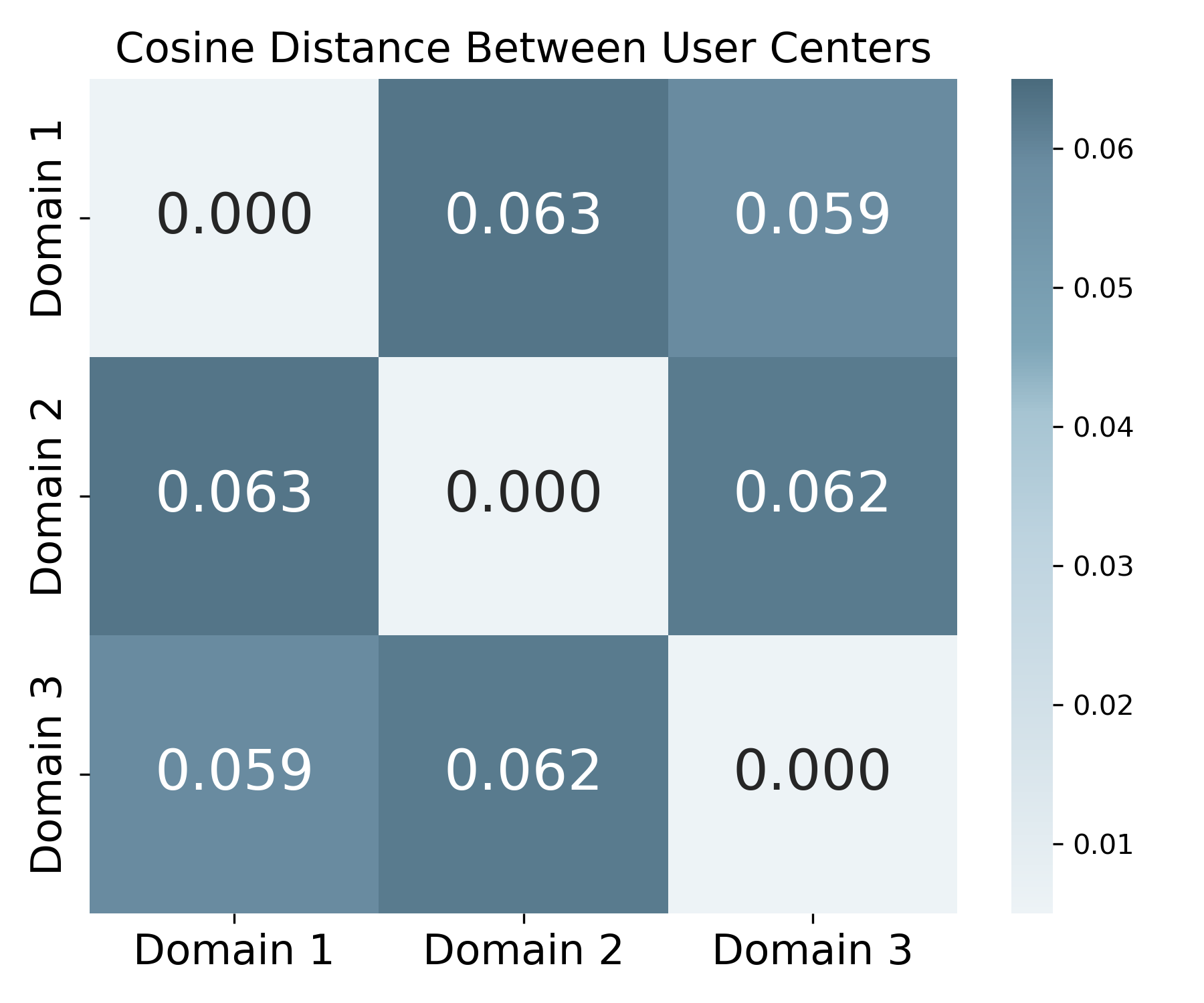}
        % \caption{Cosine Distance Between User Centers}
    
    \end{subfigure}
    \hfill
    \begin{subfigure}[b]{0.155\textwidth}
        \centering
        \includegraphics[width=\linewidth]{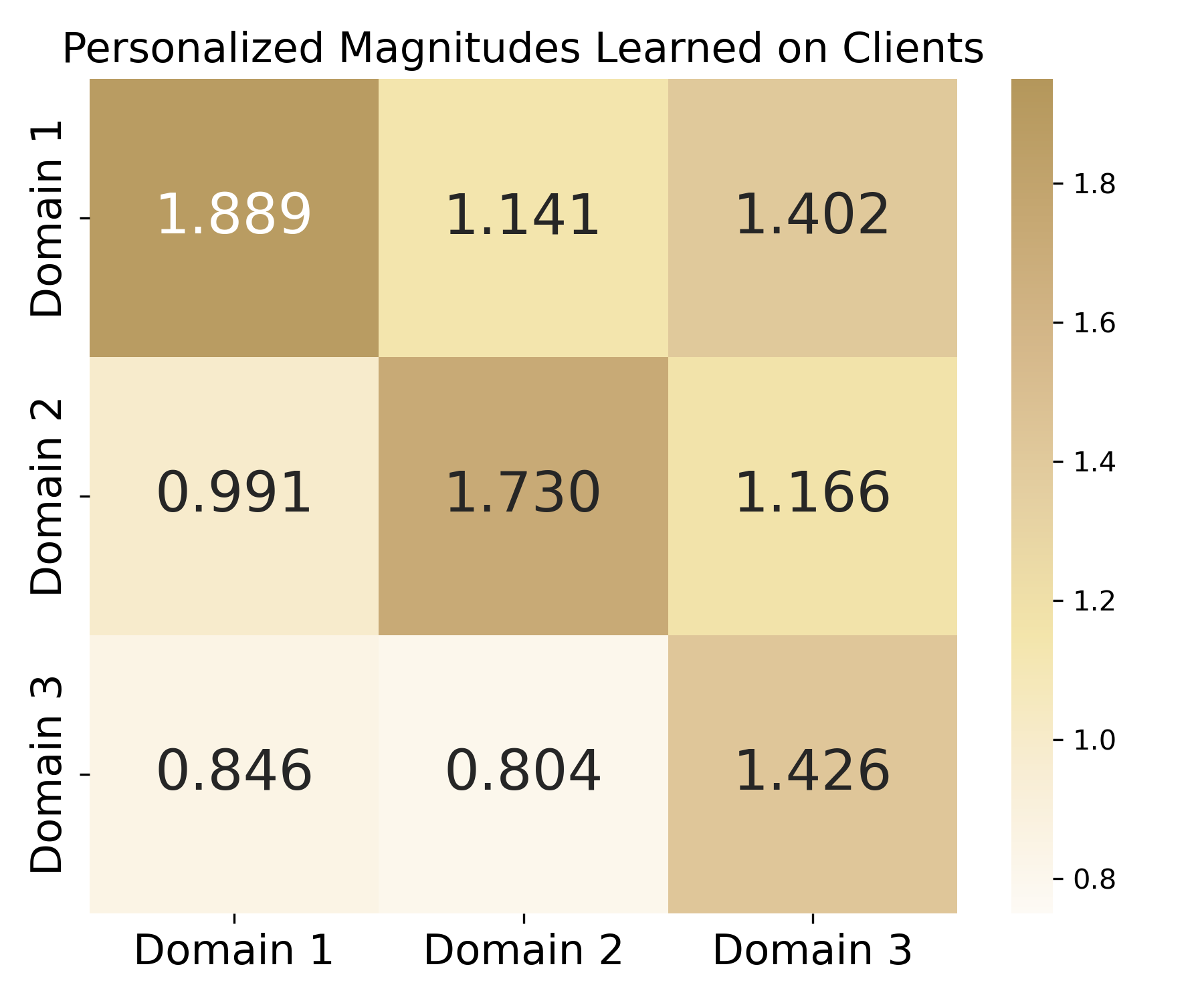}
        % \caption{Personalization Weights \(\alpha\) Learned by Each Domain (Client)}
       
    \end{subfigure}
    \vspace{-15pt}
    \caption{Visualization of item (\textit{left}) and user (\textit{middle}) distances, and learned personalization weights (\textit{right}, where each row represents the personalization magnitudes learned on a specific client/domain with respect to other domains) across domains (clients).}
    \label{fig:domain-similarity}
    \vspace{-15pt}
\end{figure}

\paragraph{Interpretability via Domain-Level User and Item Similarity} We assess whether the learned personalization weights \( \{\alpha_{ij}\} \) reflect meaningful inter-client relationships by comparing them against domain-level semantic similarities on Goodreads. We use the Sentence-BERT model \texttt{gtr-t5-base} to embed textual content in each domain. For item-side comparison, we encode all item descriptions within a domain and compute their mean embedding as the domain-level representation. We then compute cosine distances for all domain pairs based on these representations. For user-side comparison, we embed the textual user IDs and compute domain-level user representations in the same manner.

We visualize three heatmaps: (1) item-level domain cosine distances (which is 1- cosine similarity), (2) user-level domain cosine distances, and (3) the learned personalization weights \( \alpha_{ij} \). As shown in the Figure \ref{fig:domain-similarity}, the learned weights align more closely with the user-level similarity structure, which is consistent with the goal of cross-domain recommendation: transferring and aligning user preference signals across domains. Notably, for domain 3, although item embeddings suggest slightly closer similarity to domain 2, the model assigns a higher weight to domain 1, in line with the user-level similarity. This suggests the model prioritizes user-level alignment in learning cross-domain aggregation weights.

\begin{table}[t]
\setlength{\abovecaptionskip}{0.05cm}
\setlength{\belowcaptionskip}{0.2cm}
\setlength{\tabcolsep}{1.5mm}
\renewcommand{\arraystretch}{1.15}
\centering
\small
\vspace{-5pt}
\caption{Performance with LoRA rank equals 16 on GoodReads Crime \& Comics \& Children.
The best results are in bold.}
\label{tab:ablation_rank}

\begin{tabular}{@{}cccccccc@{}}
\toprule
\multirow{2}{*}{\textbf{r}} & \multirow{2}{*}{\textbf{Method}} &
\multicolumn{2}{c}{\textbf{Crime}} &
\multicolumn{2}{c}{\textbf{Comics}} &
\multicolumn{2}{c}{\textbf{Children}} \\
\cmidrule(lr){3-4}\cmidrule(lr){5-6}\cmidrule(lr){7-8}
 &  & \textbf{H@5} & \textbf{N@5} & \textbf{H@5} & \textbf{N@5} & \textbf{H@5} & \textbf{N@5} \\
\midrule
\midrule
\multirow{3}{*}{16}
 & FedAvg        & 0.0262&0.0187 & 0.0654 & 0.0532 & 0.0430 &0.0342 \\
 & PFedAvg       &0.0330 & 0.0230 & 0.0844& 0.0676 & 0.0504 & 0.0362 \\
 % & Ditto       & & & &  & &  \\
 & \textbf{\model} & \textbf{0.0346} & \textbf{0.0256} & \textbf{0.0868} & \textbf{0.0691} & \textbf{0.0514} & \textbf{0.0379} \\

\bottomrule
\end{tabular}
\vspace{-15pt}
\end{table}

\begin{table*}[h]

\setlength{\abovecaptionskip}{0.05cm}
\setlength{\belowcaptionskip}{0.2cm}
\caption{Overall performance comparison between the ID-based baselines on GoodReads Crime \& Comics \& Children. The best results are highlighted in boldface, and the second-best results are underlined.} \label{tab:classic_method}
\setlength{\tabcolsep}{3mm}{
\resizebox{\textwidth}{!}{
\begin{tabular}{c|cccc|cccc|cccc|c}
\toprule
% \hline
% \textbf{Dataset}&\multicolumn{12}{c}{\textbf{Googreads}}\\ \midrule
 \textbf{Domain} &
 %  \multicolumn{4}{c|}{} &
 %  \multicolumn{4}{c|}{} &
 %  \multicolumn{4}{c}{} \\
 % &
  \multicolumn{4}{c|}{\multirow{-1}{*}{\textbf{GoodReads- Crime}}} &
  \multicolumn{4}{c|}{\multirow{-1}{*}{\textbf{GoodReads- Comics}}} &
  \multicolumn{4}{c|}{\multirow{-1}{*}{\textbf{GoodReads- Children}}} & \multicolumn{1}{c}{\textbf{Avg}}\\
\multicolumn{1}{c|}{\textbf{Metric}} &
  \multicolumn{1}{c}{\textbf{H@5}} &
  \multicolumn{1}{c}{\textbf{H@10}} &
  \multicolumn{1}{c}{\textbf{N@5}} &
  \multicolumn{1}{c|}{\textbf{N@10}} &
  \multicolumn{1}{c}{\textbf{H@5}} &
  \multicolumn{1}{c}{\textbf{H@10}} &
  \multicolumn{1}{c}{\textbf{N@5}} &
  \multicolumn{1}{c|}{\textbf{N@10}} &
  \multicolumn{1}{c}{\textbf{H@5}} &
  \multicolumn{1}{c}{\textbf{H@10}} &
  \multicolumn{1}{c}{\textbf{N@5}} &
  \multicolumn{1}{c}{\textbf{N@10}} &
  \multicolumn{1}{c}{\textbf{H@5}}\\   \midrule\midrule
FedHCDR &
  0.0224 &
  0.0338 &
  0.0175 &
  0.0212 &
  \underline{0.0196} &
  \underline{0.0404} &
  \underline{0.0119} &
  \underline{0.0185} &
  0.0312 &
  0.0460 &
  0.0200 &
  0.0248 & 0.0244
  \\
FedDCSR &
  0.0178 &
  0.0254 &
  0.0104 &
  0.0127 &
  0.0128 &
  0.0308 &
  0.0063 &
  0.0121 &
  0.0152 &
  0.0354 &
  0.0080 &
  0.0147 & 0.0153
  \\
FedGCDR &
  \underline{0.0310} &
  \underline{0.0444} &
  \underline{0.0218} &
  \underline{0.0259} &
  0.0192 &
  0.0310 &
  0.0112 &
  0.0150 &
  \underline{0.0336} &
  \underline{0.0618} &
  \underline{0.0204} &
  \underline{0.0288} & \underline{0.0279}
  \\
  \midrule
\model &
\textbf{0.0352} &
\textbf{0.0478} &
\textbf{0.0252} &
\textbf{0.0292} &
\textbf{0.0862} &
  \textbf{0.1078} &
  \textbf{0.0701} &
  \textbf{0.0771} &
  \textbf{0.0502} &
  \textbf{0.0620} &
\textbf{0.0401} &
\textbf{0.0439} &
\textbf{0.0572}\\  \bottomrule
\end{tabular}
}}
\end{table*}

\paragraph{Convergence of Personalized Weights}

\begin{figure*}[ht]
    \centering
    \begin{subfigure}[b]{0.3\textwidth}
        \centering
        \includegraphics[width=\linewidth]{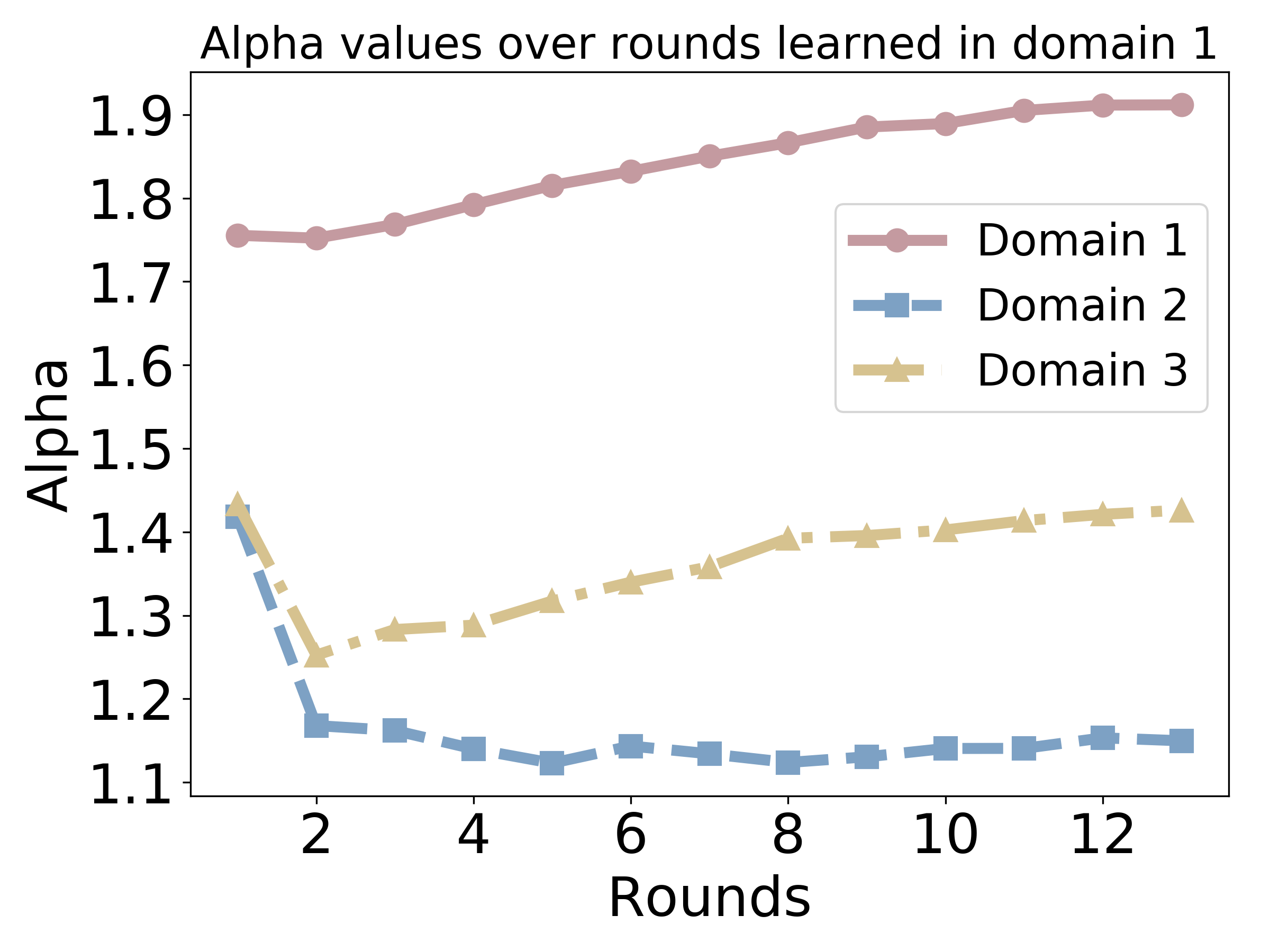}

    \end{subfigure}
    \hfill
    \begin{subfigure}[b]{0.3\textwidth}
        \centering
        \includegraphics[width=\linewidth]{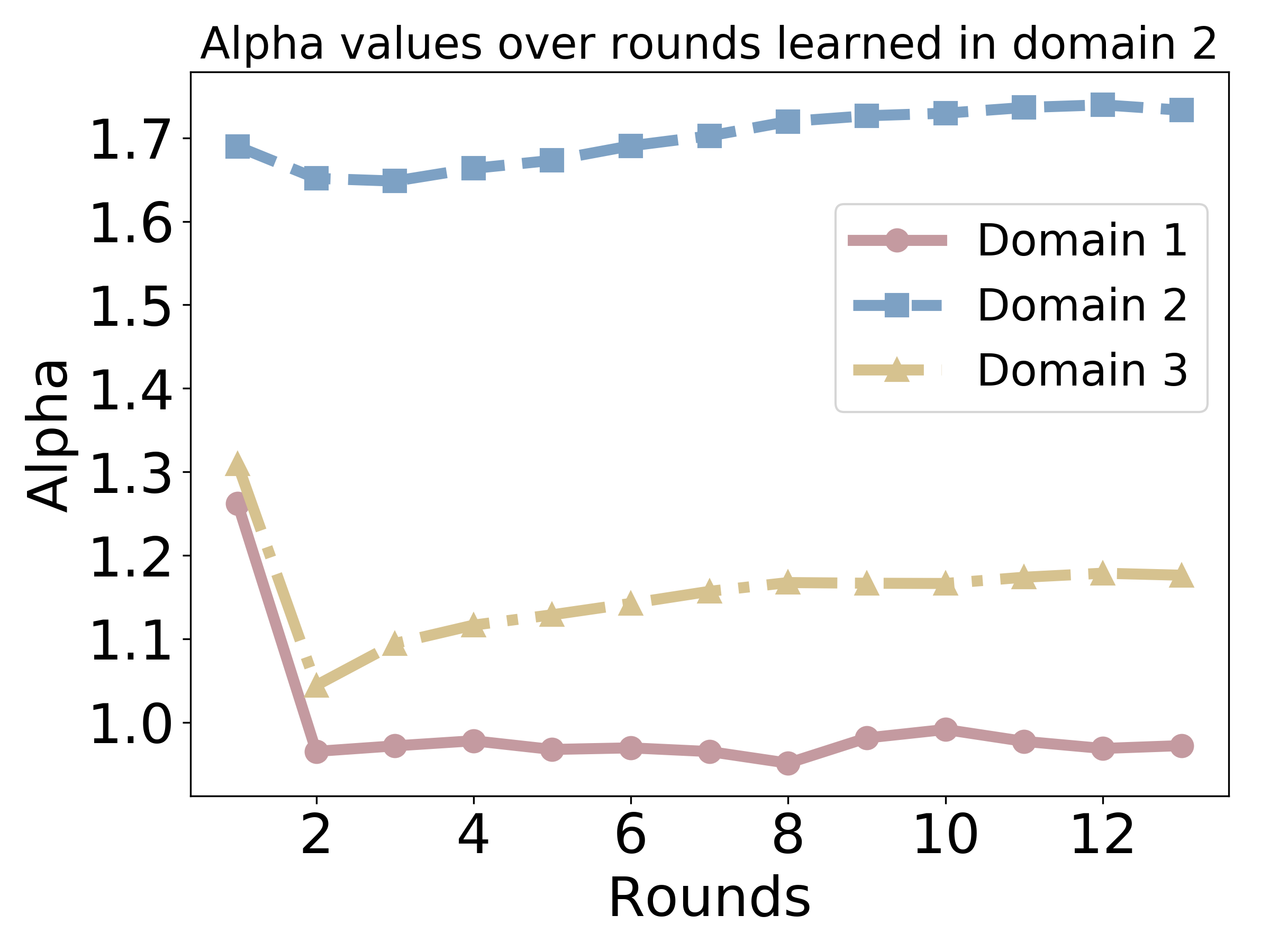}
     
    \end{subfigure}
    \hfill
    \begin{subfigure}[b]{0.3\textwidth}
        \centering
        \includegraphics[width=\linewidth]{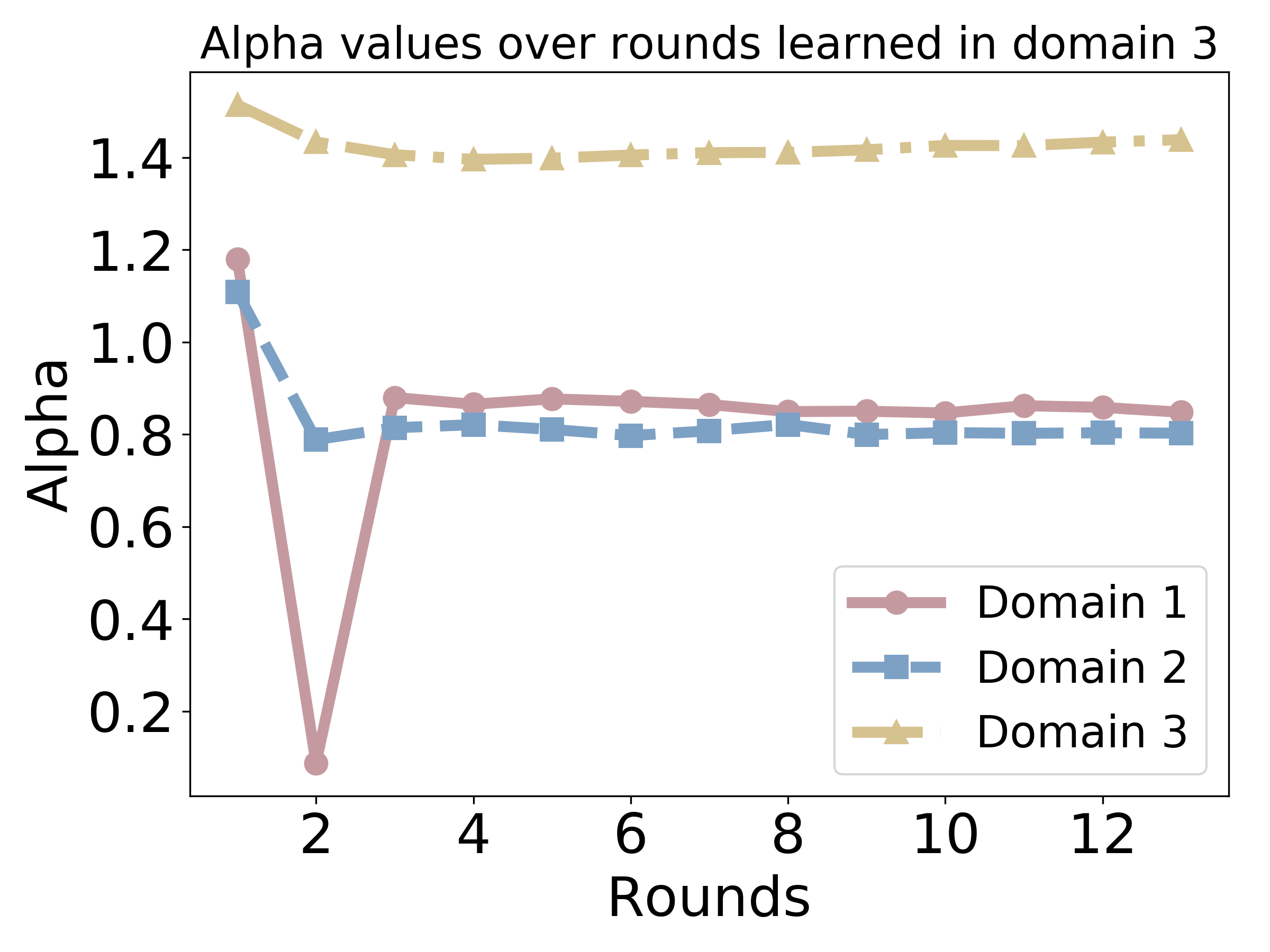}

    \end{subfigure}
    \vspace{-10pt}
    \caption{Convergence of personalized weights learned on each clients (domains) on GoodReads.}
    \label{fig:alpha_convergence}
    \vspace{-10pt}
\end{figure*}

To better understand the behavior of the personalized aggregation mechanism in \model, we examine the convergence of the learned client-specific weights $\{\alpha_{ij}\}$. As shown in Figure~\ref{fig:alpha_convergence}, we visualize the weight evolution over training rounds for three clients. We observe a clear convergence, where each clients consistently learn to promote weights from similar domains and down-weight dissimilar ones. The best performance occurs at communication round 11.

\subsubsection{Performance under Different LoRA Rank.}
In our main experiments, we use a LoRA rank of 8 to achieve more efficient communication. To further evaluate the impact of the LoRA rank $r$ on model performance, we conduct a rank study on the GoodReads \textit{Crime}, \textit{Comics}, and \textit{Children} domains. As shown in Table~\ref{tab:ablation_rank}, \model~consistently outperforms FedAvg and PFedAvg across all domains in both Hit@5 and NDCG@5 metrics when using different $r=16$, which shows the stable performance of \model.

\subsubsection{Performance Across User Activity and Item Popularity Groups (RQ4)}

To analyze model performance under varying user activity levels, we conduct a bucket-based evaluation across all three domains. Specifically, users on each client are ranked by their interaction counts and are evenly divided into five buckets, with each bucket containing 20\% of the users. Bucket 1 represents the least active users, while Bucket 5 contains the most active ones. We report the Hit@5 performance of \model~ and one of the strong baselines PFedAvg within these groups to assess how well the models generalize across different user activity levels. As shown in Figure~\ref{fig:user_bucket}, \model~outperforms PFedAvg in 13 out of 15 user buckets, and with one bucket showing equal performance, highlighting the robustness of \model{} across varying user activity levels.

In addition to user-side analysis, we evaluate model performance with respect to item popularity. Items are grouped into two groups based on their frequency in ground-truth interactions across all clients. Items in the bottom 2\% of all ground-truth items ranked by interaction frequency are categorized as inactive, while the remaining are active. This split reflects the real-world challenge of long-tail items, which are harder to recommend due to limited exposure. As shown in Figure~\ref{fig:domain_user_item}, \model~consistently outperforms PFedAvg across active and inactive item groups in all three domains, which demonstrates \model's ability to generalize to infrequent items by effectively leveraging cross-domain knowledge.

\subsubsection{Comparison with ID-based Methods}
\label{sec:id_baselines}
In this section, we compare \model~with traditional ID-based federated cross-domain recommendation methods. FedHCDR \cite{zhang2024fedhcdr} uses hypergraph filters to separate domain-shared and domain-specific user representations, enhanced by a contrastive learning module. FedDCSR \cite{zhang2024feddcsr} disentangles user sequences into shared and exclusive features and applies contrastive learning to strengthen preference signals. FedGCDR \cite{yang2024federated} extracts transferable knowledge via graph attention networks and expands the target domain graph using a feature mapping mechanism. For a fair comparison, the embedding sizes are set to 128 for all baseline methods. The results in Table~\ref{tab:classic_method} show that leveraging side information through LLMs leads to improved performance on federated cross-domain recommendation tasks.

\subsection{Privacy \& Efficiency Analysis}

\subsubsection{Privacy Preserving Analysis}

\label{privacy_analysis}
In the proposed \model, privacy is enhanced through the use of \emph{directional components} of LoRA updates as shared parameters. Specifically, after each client completes its local training, it uploads raw LoRA matrices \( A_j \) and \( B_j \) to the server. Instead of sharing these updates directly, the server performs normalization to extract \emph{directional components} \( \tilde{A}_j \) and \( \tilde{B}_j \), which are then broadcast to all clients. This removes scale information, acting as an inherent obfuscation mechanism.

For stronger protection, normalization can be performed locally before upload so the server never accesses raw $(A_j,B_j)$. Moreover, \model\ is compatible with local differential privacy \cite{8416855}: clients can apply an $\varepsilon$-LDP perturbation to the communicated directional components before uploading, providing an explicit privacy guarantee.
\vspace{-5pt}

\begin{table}[h]
\centering
\small
\caption{Comparison of communication cost per client per round}
\vspace{-10pt}
\scalebox{0.97}{
\begin{tabular}{lccc}
\toprule
\textbf{Method} & \textbf{Upload Cost} & \textbf{Download Cost} & \textbf{Total Cost} \\
\midrule
PFedAvg        & \( \mathcal{O}(rd) \)       & \( \mathcal{O}(rd) \)            & \( \mathcal{O}(rd) \) \\

Ditto        & \( \mathcal{O}(rd) \)       & \( \mathcal{O}(rd) \)            & \( \mathcal{O}(rd) \) \\

FDLoRA        & \( \mathcal{O}(rd) \)       & \( \mathcal{O}(rd) \)            & \( \mathcal{O}(rd) \) \\
FellRec        & \( \mathcal{O}(bh) \) & \( \mathcal{O}(bh) \)    & \( \mathcal{O}(2bh) \) \\
\model~  & \( \mathcal{O}(rd) \)     & \( \mathcal{O}(Krd) \)         & \( \mathcal{O}(Krd) \) \\
\bottomrule
\end{tabular}}
\label{tab:comm_cost}
\vspace{-10pt}
\end{table}

\subsubsection{Communication Cost Analysis}
\label{communication_cost}
In LLM-based federated recommendation systems, communication cost is a major bottleneck, especially when model sizes are large and communication rounds are frequent. We analyze the communication cost of \model~in comparison with several representative baselines: PFedAvg, FellRec, Ditto, FDLoRA, and FedllRec. In each communication round, the communication overhead is determined primarily by exchanging parameters between clients and servers. Let \( d \) denote the dimensionality of a full model layer, and \( r \ll d \) denote the rank of the LoRA update. In \model~, each client uploads its own LoRA matrices \( A_i, B_i \), and downloads the normalized directional components \( \tilde{A}_j, \tilde{B}_j \) from other clients. Although this introduces a slightly higher download cost compared to typical LoRA-based methods, it enables fully personalized aggregation. The per-round communication cost is \( \mathcal{O}(K rd) \), where \( K \) is the number of clients.

\begin{figure*}[ht]
    \centering
    \begin{subfigure}[b]{0.32\textwidth}
        \centering
        \includegraphics[width=\linewidth]{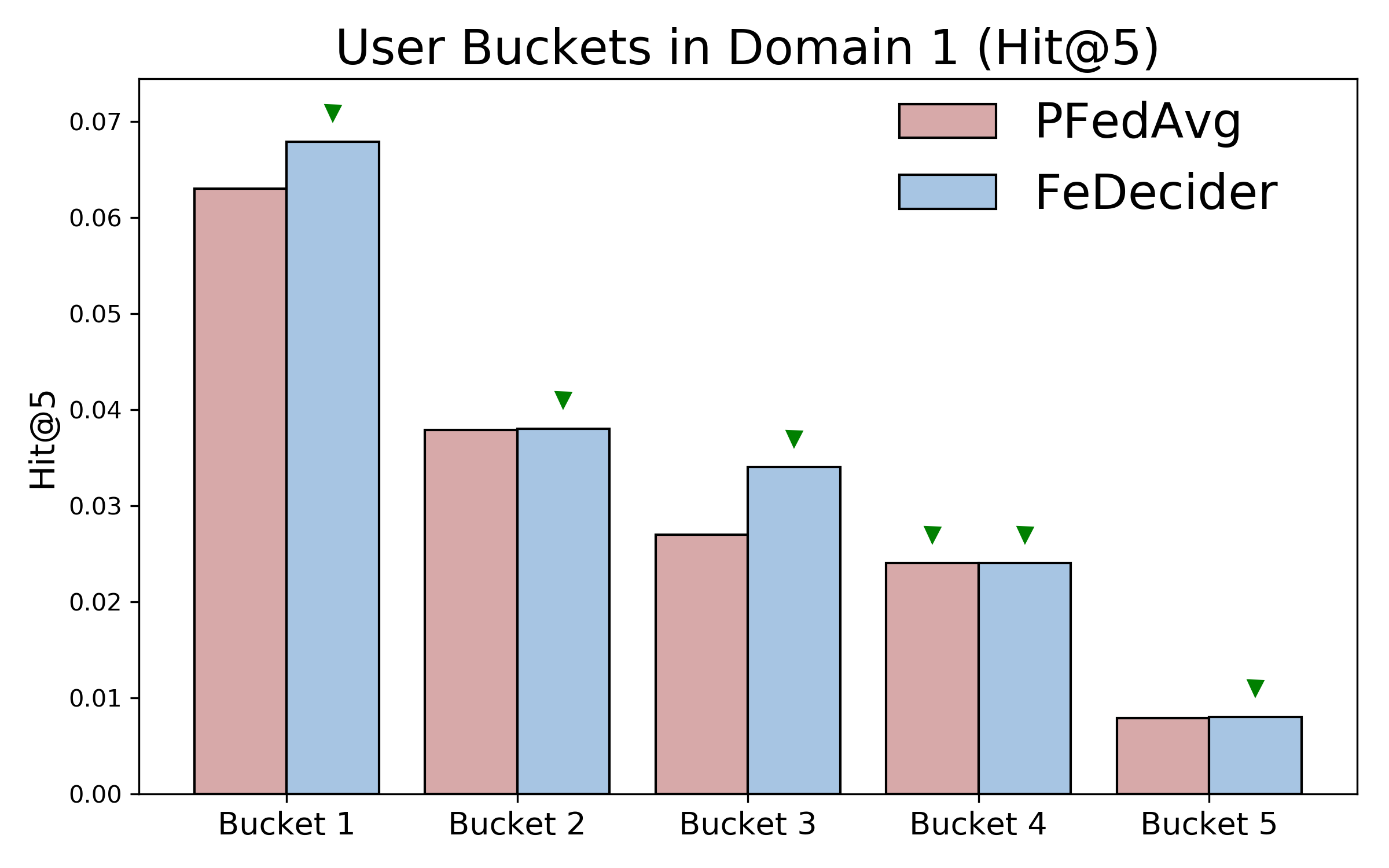}
        % \caption{Cosine Distance Between Item Centers}
      
    \end{subfigure}
    \hfill
    \begin{subfigure}[b]{0.32\textwidth}
        \centering
        \includegraphics[width=\linewidth]{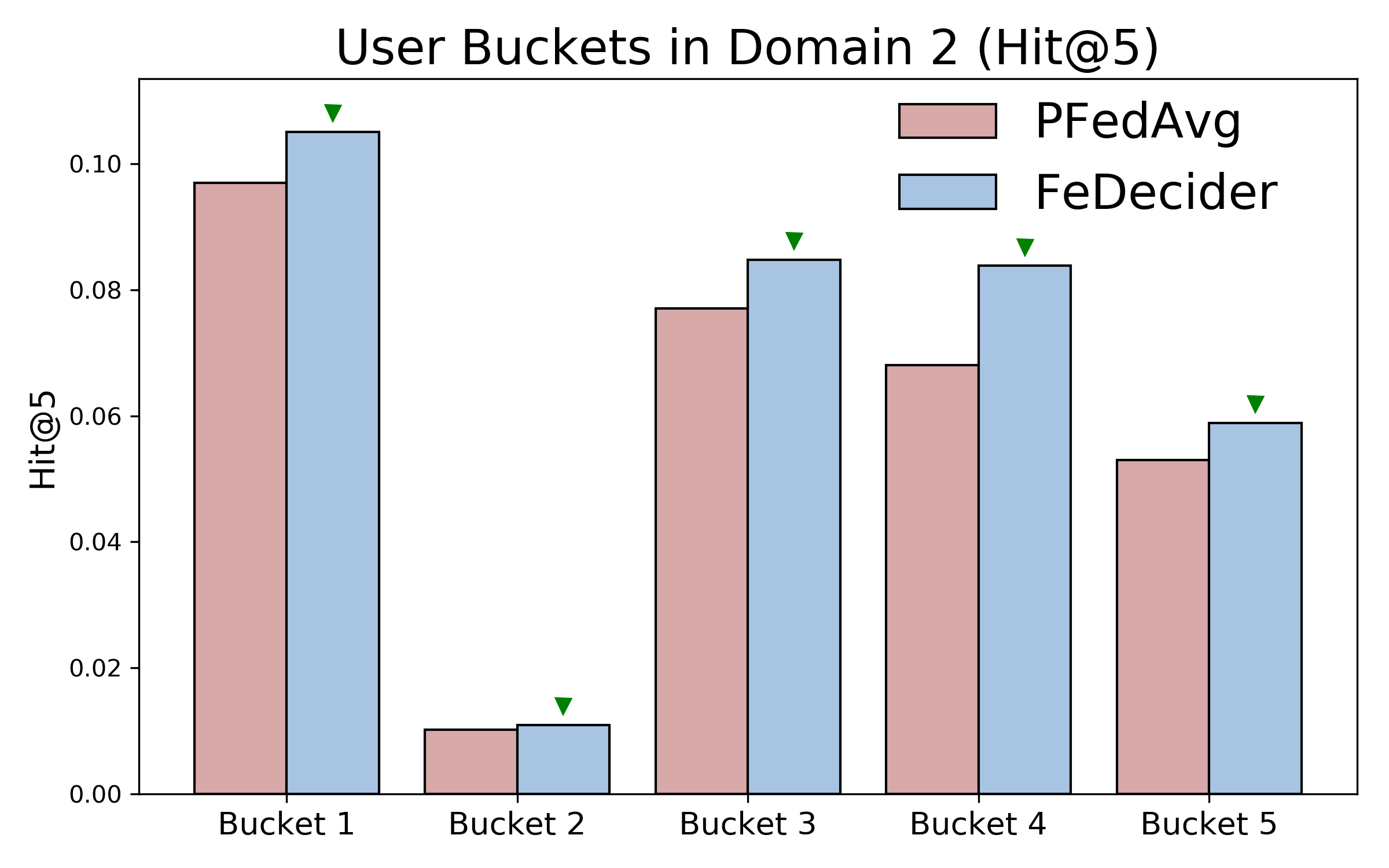}
        % \caption{Cosine Distance Between User Centers}
      
    \end{subfigure}
    \hfill
    \begin{subfigure}[b]{0.32\textwidth}
        \centering
        \includegraphics[width=\linewidth]{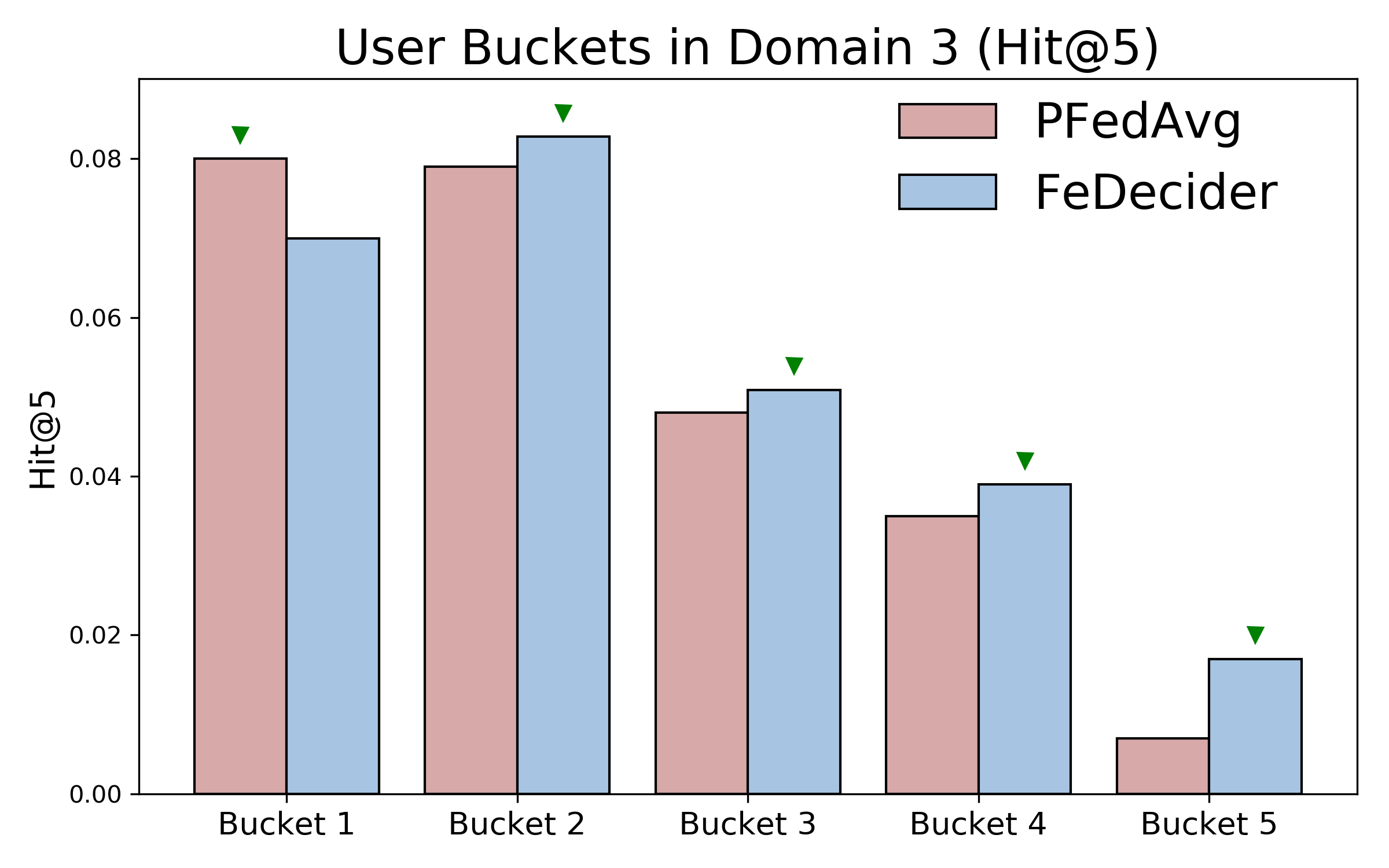}
        % \caption{Personalization Weights \(\alpha\) Learned by Each Domain (Client)}
      
    \end{subfigure}
    \vspace{-10pt}
    \caption{Comparison of Hit@5 across user buckets in domains, with check marks indicating the better-performing method.}
    \label{fig:user_bucket}
    \vspace{-5pt}
\end{figure*}

\begin{figure*}[ht]
    \centering

    \includegraphics[width=0.161\textwidth]{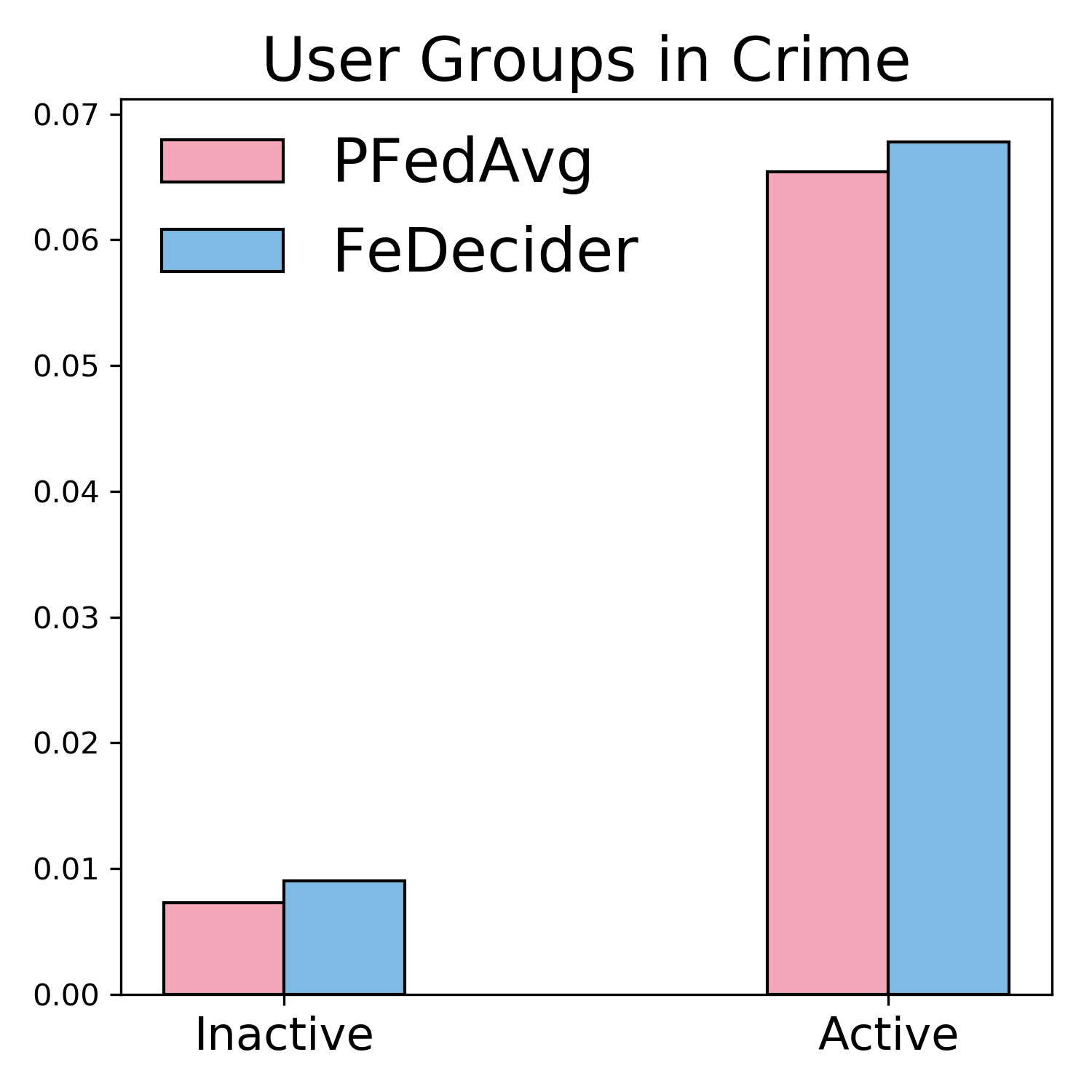}
    \includegraphics[width=0.161\textwidth]{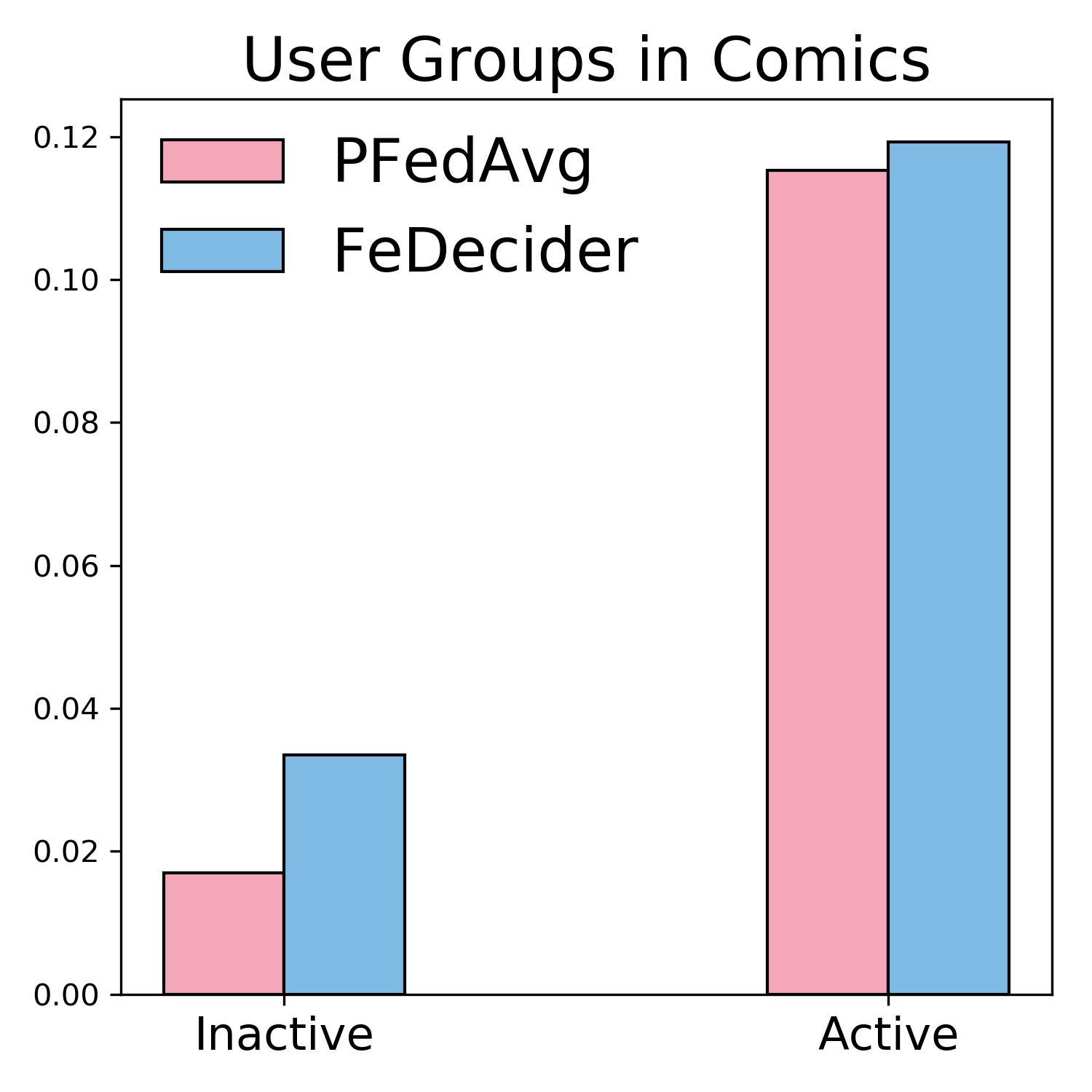}
    \includegraphics[width=0.161\textwidth]{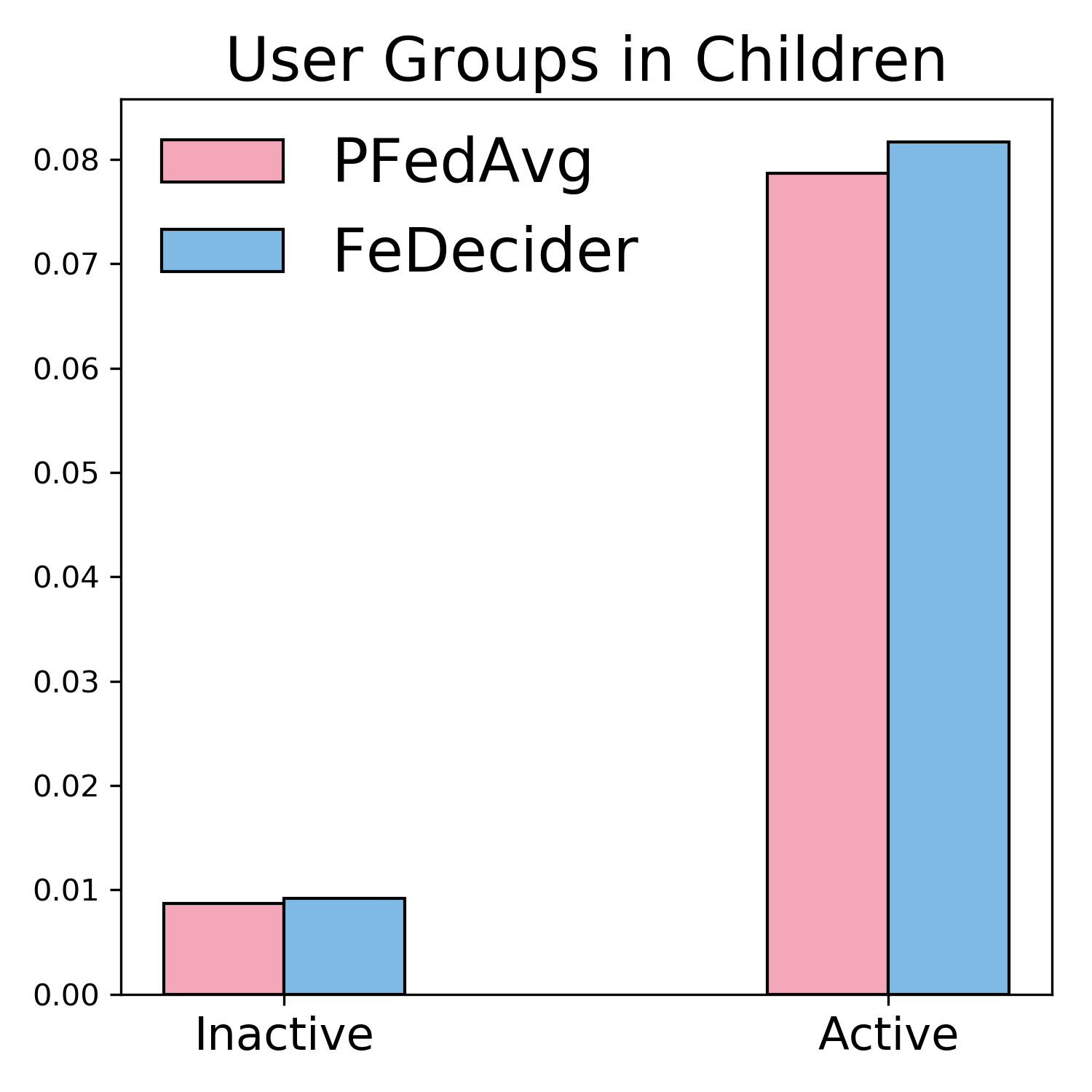}
    \includegraphics[width=0.161\textwidth]{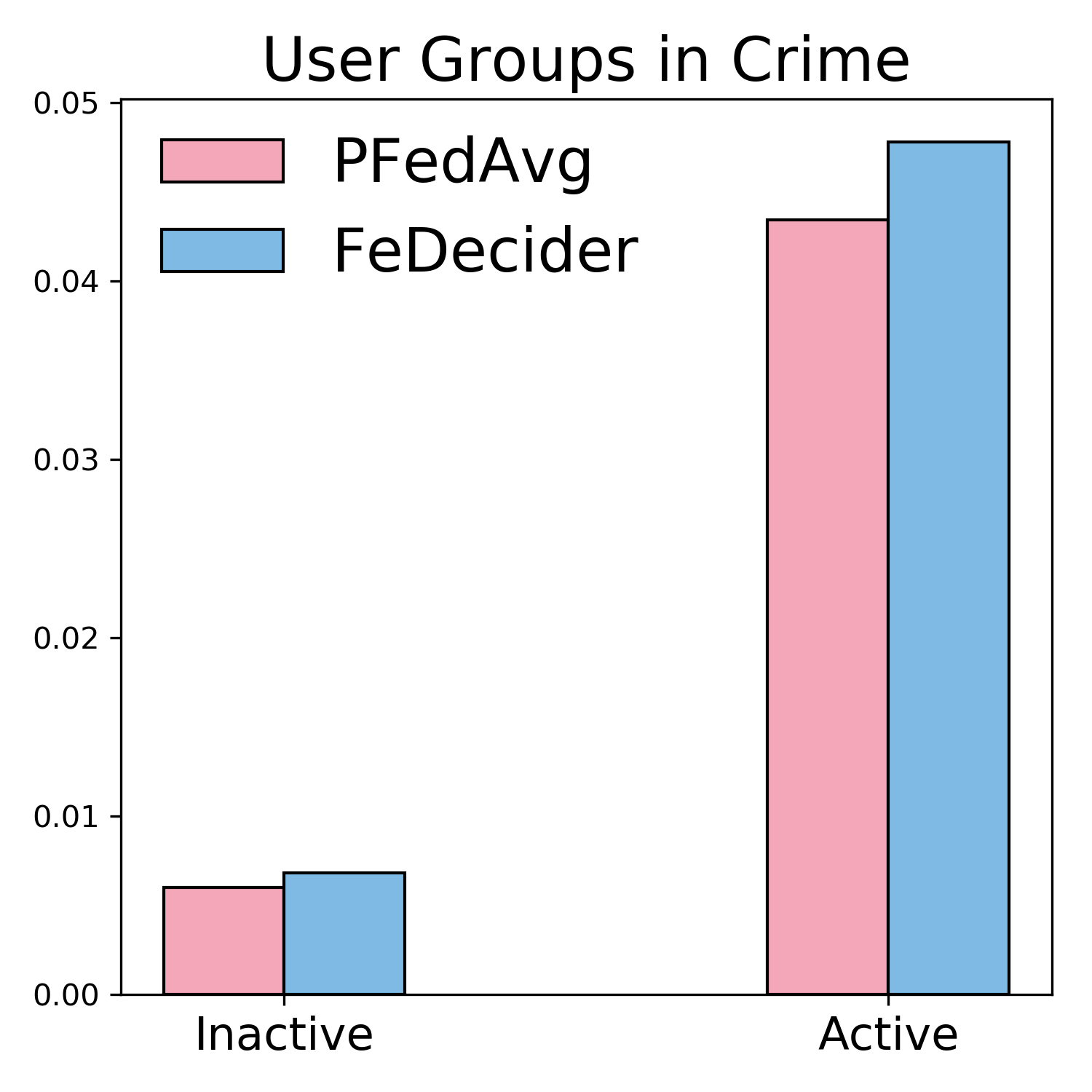}
    \includegraphics[width=0.161\textwidth]{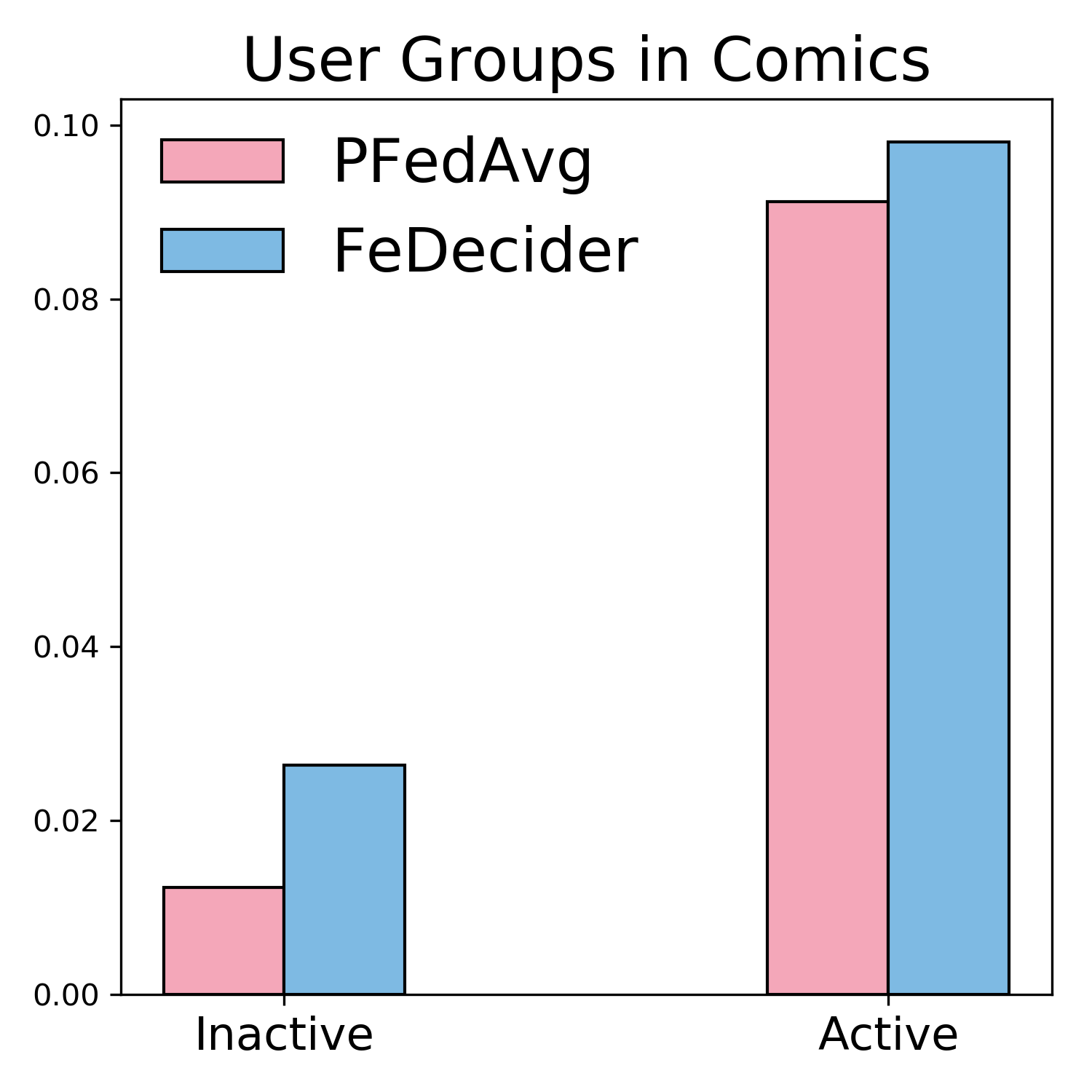}
    \includegraphics[width=0.161\textwidth]{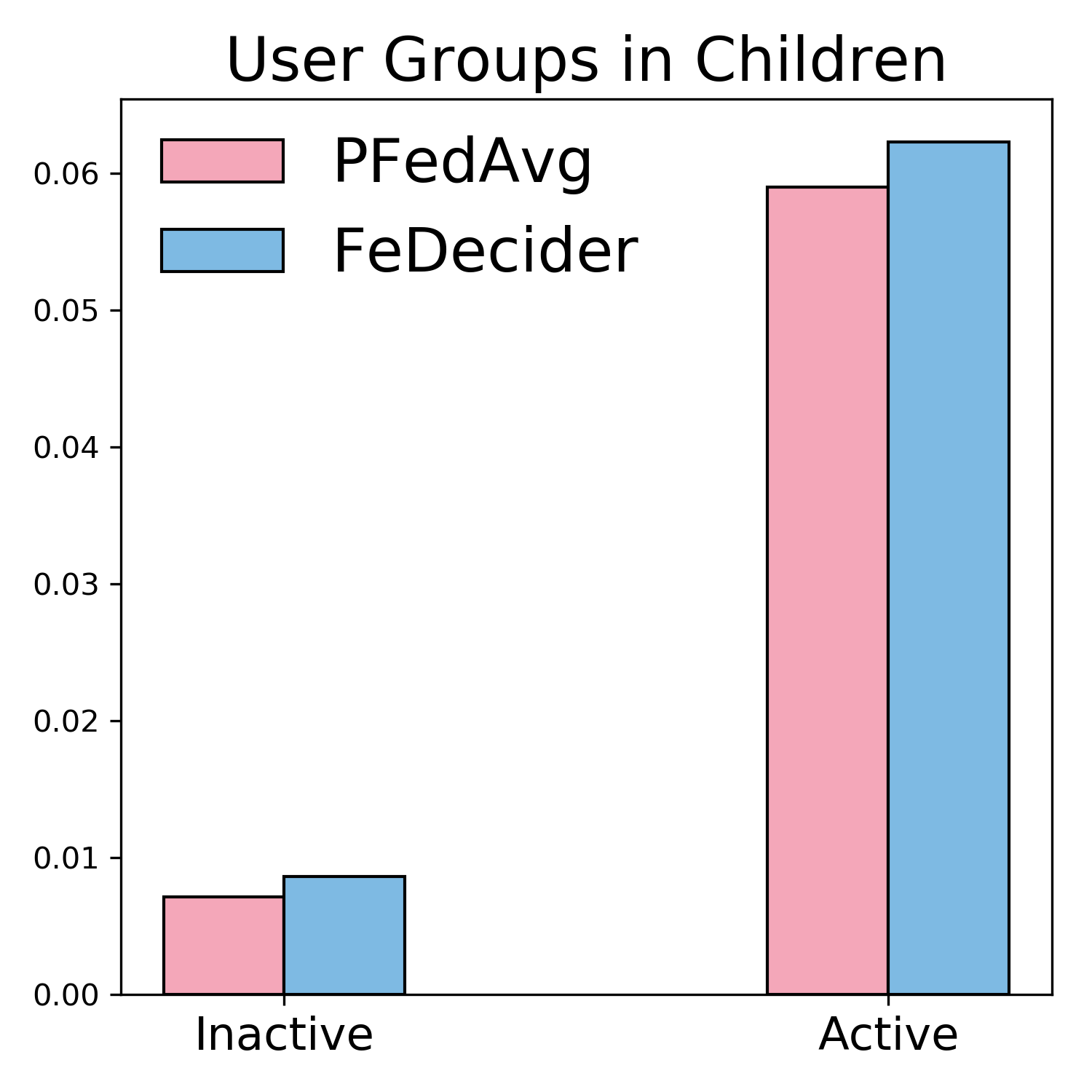}
    \vspace{-10pt}
    \caption{Performance comparison on item groups in three domains. The left three charts show Hit@5, and the right three show NDCG@5.}
    \label{fig:domain_user_item}
    \vspace{-10pt}
\end{figure*}

Compared to other LoRA-based methods (Table~\ref{tab:comm_cost}), \model~incurs a slightly higher download cost due to the need to download shared directional components from all other clients for personalized aggregation. owever, this additional cost is minimal in practice, since federated cross-domain recommendation typically involves only a small number of closely related domains/clients (i.e., \(K\) is small), as incorporating many weakly related domains often brings negative transfer with limited benefit. Overall, \model~achieves a favorable trade-off between communication efficiency and effectiveness in Federated CDR. We provide a discussion of the training time and convergence in Appendix \ref{appendix:computation}.

\vspace{-5pt}
\section{Related Work}
\label{sec:related_works}
\paragraph{Cross-domain Recommendation (CDR) \& Federated CDR} 
Cross-domain recommendation alleviates data sparsity problem by transferring knowledge across related domains \cite{khan2017cross, chen2022differential,10.1145/3640457.3688145}. Early methods like SSCDR \cite{kang2019semi}, EMCDR \cite{man2017cross} learns mapping functions across embeddings, while TMCDR \cite{zhu2021transfer} applies meta-learning for cold-start users. Graph-based models like PPGN \cite{zhao2019cross} and Coast \cite{zhao2023cross} enhance representation learning via structural and contrastive signals.

To enable privacy and personalization, recent works integrate cross-domain recommendation with federated learning. FedCT \cite{liu2021fedct} trains decentralized user embeddings via a federated VAE. FedCDR \cite{yan2022fedcdr} introduces personalized transfer modules. PriCDR \cite{chen2022differential} protects user interactions using random projections such as the Johnson-Lindenstrauss transform. P2FCDR \cite{chen2023win} adopts gating-based dual-target optimization to avoid negative transfer. FedGCDR \cite{yang2024federated} and FedCLR \cite{wang2025federated} further improve knowledge transfer via privacy-preserving mechanisms and contrastive alignment. Despite these efforts, most methods still rely on ID-based representations and traditional recommendation architectures. Few have explored using large language models for federated cross-domain recommendation, which offers new opportunities for semantic-rich and privacy-aware knowledge transfer.

\vspace{-10pt}
\paragraph{LLM-based Federated Learning}

With the success of large language models \cite{achiam2023gpt, liu2024deepseek,he-etal-2025-llm,Zou2025GTRGF}, privacy concerns have raised attention in federated learning over decentralized data \cite{illman2019california}. Early work like FedBERT \cite{tian2022fedbert} explores federated pretraining on BERT \cite{devlin2019bert}. As models grow, recent efforts adopt parameter-efficient finetuning \cite{Ai_2025,ai2025nirvanastructuredpruningreimagined}, notably LoRA \cite{hu2022lora}, to reduce memory and communication costs \cite{ding2023parameter, babakniya2023slora}. FLoRA \cite{wang2024flora} stacks local adapters into global modules, while FFA-LoRA \cite{sun2024improving} freezes initialized weights to cut training overhead. RoLoRA \cite{chen2024robust} improves stability via alternating aggregation. FedLoRA \cite{wu2024fedlora} addresses non-IID data but requires training from scratch, limiting scalability. To improve personalization, FDLoRA \cite{qi2024fdlora} uses dual adapters for global and local knowledge, and FellRec \cite{zhao2025federated} applies LLM-based federated learning in general recommendation scenario. However, little work has explored federated LLM finetuning for cross-domain recommendation, where item and user shifts pose unique challenges for shared knowledge integration and fine-grained personalization under privacy constraints.

\vspace{-5pt}
\section{Conclusion}
\label{sec:conclusion}

We propose \model, an LLM-based framework for federated cross-domain recommendation. \model~addresses key challenges in adopting LLM-based recommendation models in Federated CDR. Inspired by recent insights into the structure of Low-Rank updates, \model~computes directional components from each client's Low-Rank updates and transfers them across clients to mitigate the negative impact of the magnitude (scale) of the domain-specific adapters. Instead of performing aggregation on the server, \model~shifts this process to the client side, where each client performs data-aware integration of others' directional components using a set of client-specific learnable weights. This design allows clients to adaptively incorporate shared knowledge based on their own data distributions. We demonstrate the effectiveness and robustness of \model~through extensive experiments on cross-domain recommendation datasets. Future work includes extending the personalized weights to a layer-wise design and integrating image information for VLM-based multimodal CDR.

\section{Acknowledgements}
This work is supported by National Science Foundation under Award No. IIS-2316233, IIS-2117902. The views and conclusions are those of the authors and should not be interpreted as representing the official policies of the funding agencies or the government.

% \begin{acks}
% \end{acks}

\bibliographystyle{ACM-Reference-Format}
\bibliography{main}

\appendix
\section{Theoretical Analysis}

In this section, we present an analysis of the proposed decomposition into shared directions and personalized weights. We first introduce the notation, and then give several lemmas and propositions together with short proofs.

\subsection{Preliminaries}

Let $W_0$ denote the global backbone parameters shared by all clients. For each client $i \in \{1,\dots,N\}$, we write its personalized model as
\begin{equation}
    W_i \;=\; W_0 + \Delta W_i,
    \qquad
    \Delta W_i \;=\; \sum_{j=1}^N \alpha_{ij} D_j,
    \label{eq:deltaW-decomposition}
\end{equation}
where $D_j = \tilde{B}_j \tilde{A}_j$ is the normalized LoRA direction from client $j$ and satisfies
\begin{equation}
    \|D_j\|_F = 1, \qquad j=1,\dots,N,
\end{equation}
and $\alpha_{ij} \in \mathbb{R}$ is the personalized aggregation weight of direction $D_j$ for client $i$.

The local objective of client $i$ is
\begin{equation}
    F_i(W_i) \;=\; \mathbb{E}_{(x,y)\sim \mathcal{D}_i}
    \big[\,\mathcal{L}(f(x;W_i),y)\,\big],
\end{equation}
and its training problem under the above parameterization is
\begin{equation}
    \min_{\{\alpha_{ij}\}}
    F_i\Big(W_0 + \sum_{j=1}^N \alpha_{ij} D_j\Big).
    \label{eq:local-objective-alpha}
\end{equation}

We call the subspace
\begin{equation}
    \mathcal{S} \;:=\; \mathrm{span}\{D_1,\dots,D_N\}
\end{equation}
the shared low-rank update subspace.

\begin{definition}[Harmful and beneficial directions]
For client $i$, define the first-order coefficient of direction $D_j$ at $W_0$ by
\begin{equation}
    g_{ij} \;:=\;
    \big\langle \nabla F_i(W_0),\, D_j \big\rangle_F,
\end{equation}
where $\langle A,B\rangle_F = \mathrm{tr}(A^\top B)$ is the Frobenius inner product.

We say $D_j$ is \emph{beneficial} for client $i$ if $g_{ij} < 0$, and \emph{harmful} for client $i$ if $g_{ij} > 0$.
\end{definition}

The following first-order approximation is standard.

\begin{lemma}[First-order expansion]\label{lem:first-order}
Let $F_i$ be differentiable. For any $\Delta W \in \mathbb{R}^{d\times d}$, we have
\begin{equation}
    F_i(W_0 + \Delta W)
    \;=\;
    F_i(W_0)
    + \big\langle \nabla F_i(W_0),\, \Delta W \big\rangle_F
    + O\big(\|\Delta W\|_F^2\big).
\end{equation}
In particular, if $\Delta W_i$ is given by \eqref{eq:deltaW-decomposition}, then
\begin{equation}
    F_i(W_0 + \Delta W_i)
    =
    F_i(W_0)
    + \sum_{j=1}^N \alpha_{ij} g_{ij}
    + O\Big(\Big\|\sum_{j=1}^N \alpha_{ij} D_j\Big\|_F^2\Big).
    \label{eq:first-order-alpha}
\end{equation}
\end{lemma}

\begin{proof}
This is the standard first-order Taylor expansion of a differentiable function around $W_0$.
\end{proof}

\subsection{Directional Decomposition and Negative Transfer}

For comparison, consider a baseline scheme that directly aggregates the raw LoRA updates $\Delta W_j^{\mathrm{raw}}$ from all clients on the server and broadcasts the same aggregated update to all clients, namely
\begin{equation}
    \Delta W_i^{\mathrm{base}}
    \;=\;
    \sum_{j=1}^N \beta_j \Delta W_j^{\mathrm{raw}},
    \qquad \text{for all } i,
    \label{eq:baseline-agg}
\end{equation}
where $\beta_j$ are aggregation coefficients shared by all clients.

% \begin{remark}[Effect of normalization and separating scale]
For each client $j$, write its raw LoRA update as
\[
    \Delta W_j^{\mathrm{raw}} = s_j  D_j,
    \qquad
    s_j := \|\Delta W_j^{\mathrm{raw}}\|_F,\quad
    \| D_j\|_F = 1.
\]
Then the baseline aggregation \eqref{eq:baseline-agg} can be rewritten as
\[
    \Delta W_i^{\mathrm{base}}
    =
    \sum_{j=1}^N (\beta_j s_j)\, D_j.
\]
Under the first-order approximation in Lemma~\ref{lem:first-order}, the contribution
of direction $D_j$ to the loss of client $i$ is proportional to
$(\beta_j s_j) g_{ij}$, where $g_{ij} = \langle \nabla F_i(W_0), D_j\rangle_F$.
Hence directions coming from clients with large update norm $s_j$ may dominate
the aggregated update on \emph{all} clients, even when $g_{ij} > 0$ and the
direction is harmful for client $i$.

In contrast, the proposed decomposition \eqref{eq:deltaW-decomposition} replaces
$\{\Delta W_j^{\mathrm{raw}}\}$ by normalized directions $D_j$ with
$\|D_j\|_F = 1$ and client-specific coefficients $\alpha_{ij}$. The first-order
term for client $i$ becomes
\[
    F_i(W_0 + \Delta W_i)
    \approx
    F_i(W_0)
    + \sum_{j=1}^N \alpha_{ij} g_{ij},
\]
so the influence of direction $D_j$ on client $i$ is determined solely by the
alignment $g_{ij}$ and its personalized weight $\alpha_{ij}$, rather than by the
raw scale $s_j$ of other clients. This decoupling of direction and scale
prevents large-magnitude but harmful updates from dominating all clients.

We now formalize the representational difference between our decomposition
\eqref{eq:deltaW-decomposition} and the baseline form \eqref{eq:baseline-agg}.

\begin{lemma}[Representation in the shared subspace]\label{lem:representation}
Let $\{D_j\}_{j=1}^N$ be linearly independent. Then for any $\Delta W \in \mathcal{S}$ there exists a unique coefficient vector $\alpha = (\alpha_1,\dots,\alpha_N)$ such that
\begin{equation}
    \Delta W = \sum_{j=1}^N \alpha_j D_j.
\end{equation}
\end{lemma}

\begin{proof}
Since $\{D_j\}_{j=1}^N$ are linearly independent, they form a basis of $\mathcal{S}$. Uniqueness and existence of the coordinates in this basis follow from basic linear algebra.
\end{proof}

\begin{proposition}[Personalization capacity]\label{prop:personalization}
Suppose $\{D_j\}_{j=1}^N$ are linearly independent and let $\mathcal{S}$ be their span. Consider a family of desired personalized updates $\{\Delta W_i^\star\}_{i=1}^N \subset \mathcal{S}$.

\begin{enumerate}
    \item There always exist client-specific coefficients $\alpha_i = (\alpha_{i1},\dots,\alpha_{iN})$ such that
    \begin{equation}
        \Delta W_i^\star = \sum_{j=1}^N \alpha_{ij} D_j,
        \qquad i=1,\dots,N.
    \end{equation}
    \item If there exist clients $i \neq k$ such that $\Delta W_i^\star$ and $\Delta W_k^\star$ are not colinear, i.e.,
    \begin{equation}
        \Delta W_i^\star \not\propto \Delta W_k^\star,
    \end{equation}
    then there is no single coefficient vector $\beta = (\beta_1,\dots,\beta_N)$ and client-specific scalars $c_i$ that satisfy
    \begin{equation}
        \Delta W_i^\star = c_i \sum_{j=1}^N \beta_j D_j,
        \qquad \text{for all } i.
    \end{equation}
\end{enumerate}
\end{proposition}

\begin{proof}
For (1), Lemma~\ref{lem:representation} applied to each $\Delta W_i^\star \in \mathcal{S}$ gives the existence and uniqueness of $\alpha_i$.

For (2), if such $\beta$ and $\{c_i\}$ existed, then all $\Delta W_i^\star$ would lie in the one-dimensional subspace spanned by $\sum_{j=1}^N \beta_j D_j$, hence would be pairwise colinear, contradicting the assumption that there exist $i\neq k$ with $\Delta W_i^\star \not\propto \Delta W_k^\star$.
\end{proof}

Proposition~\ref{prop:personalization} implies that the directional decomposition can realize arbitrary client-specific updates inside the shared subspace $\mathcal{S}$. In contrast, even if we grant a FedAvg-style scheme additional flexibility by allowing each client to rescale a shared aggregated direction with a scalar $c_i$, i.e., $\Delta W_i = c_i \sum_{j=1}^N \beta_j D_j$ with a single shared $\beta$, the achievable updates are still restricted to a one-dimensional family and thus can't realize heterogeneous updates across clients. This is the source of negative transfer when different clients require conflicting update directions.

\subsection{Gradient Characterization of Personalized Weights}
\label{app:personalized_weight}
We now examine how optimizing the weights $\alpha_{ij}$ affects harmful and beneficial directions.

\begin{lemma}[Gradient w.r.t.\ personalized weights]\label{lem:grad-alpha}
Let $\Delta W_i(\alpha_i) = \sum_{j=1}^N \alpha_{ij} D_j$ and $F_i$ be differentiable. Then
\begin{equation}
    \frac{\partial F_i(W_0 + \Delta W_i(\alpha_i))}{\partial \alpha_{ij}}
    =
    \big\langle
        \nabla F_i\big(W_0 + \Delta W_i(\alpha_i)\big),
        \, D_j
    \big\rangle_F.
    \label{eq:alpha-gradient}
\end{equation}
In particular, at $\alpha_i = 0$,
\begin{equation}
    \frac{\partial F_i(W_0)}{\partial \alpha_{ij}}
    \Big|_{\alpha_i = 0}
    =
    g_{ij} = \big\langle \nabla F_i(W_0),\,D_j\big\rangle_F.
\end{equation}
\end{lemma}

\begin{proof}
By the chain rule,
\begin{equation*}
    \frac{\partial F_i(W_0 + \Delta W_i(\alpha_i))}{\partial \alpha_{ij}}
    =
    \big\langle
        \nabla F_i\big(W_0 + \Delta W_i(\alpha_i)\big),
        \, \frac{\partial \Delta W_i(\alpha_i)}{\partial \alpha_{ij}}
    \big\rangle_F.
\end{equation*}
Since $\Delta W_i(\alpha_i) = \sum_{k=1}^N \alpha_{ik} D_k$, we have $\partial \Delta W_i(\alpha_i)/\partial \alpha_{ij} = D_j$, which gives \eqref{eq:alpha-gradient}. Evaluating at $\alpha_i = 0$ yields the second claim.
\end{proof}

\begin{table*}[t]

\setlength{\abovecaptionskip}{0.05cm}
\setlength{\belowcaptionskip}{0.2cm}
\caption{Overall performance comparison between the baselines on Amazon Electronics \& Phones. The best results are highlighted in boldface, and the second-best results are underlined.} \label{tab:main-add}
\setlength{\tabcolsep}{3mm}{
\begin{adjustbox}{width=0.8\textwidth,center} 
\begin{tabular}{c|cccc|cccc|c}
\toprule
% \hline
% \textbf{Dataset}&\multicolumn{12}{c}{\textbf{Googreads}}\\ \midrule
 \textbf{Domain} &
 %  \multicolumn{4}{c|}{} &
 %  \multicolumn{4}{c|}{} &
 %  \multicolumn{4}{c}{} \\
 % &
  \multicolumn{4}{c|}{\multirow{-1}{*}{\textbf{Amazon- Electronics}}} &
  \multicolumn{4}{c|}{\multirow{-1}{*}{\textbf{Amazon- Phones}}} &
  \multicolumn{1}{c}{\textbf{Avg}}\\
\multicolumn{1}{c|}{\textbf{Metric}} &
  \multicolumn{1}{c}{\textbf{H@5}} &
  \multicolumn{1}{c}{\textbf{H@10}} &
  \multicolumn{1}{c}{\textbf{N@5}} &
  \multicolumn{1}{c|}{\textbf{N@10}} &
  \multicolumn{1}{c}{\textbf{H@5}} &
  \multicolumn{1}{c}{\textbf{H@10}} &
  \multicolumn{1}{c}{\textbf{N@5}} &
  \multicolumn{1}{c|}{\textbf{N@10}} &
  \multicolumn{1}{c}{\textbf{H@5}}\\   \midrule\midrule

FedAvg &
 0.0114 &
 0.0130&
 0.0092&
0.0097 &
0.0164 &
 0.0194 &
  0.0128&
 0.0137& 0.0139
  \\
PFedAvg &
  0.0122 &
  \underline{0.0150} &
\underline{0.0095}&
  \underline{0.0105} &
  0.0186 &
  0.0248 &
  0.0139 &
  0.0159 & 0.0154
 
  \\
FedProx &
  0.0116 &
  0.0136 &
  0.0088 &
  0.0095 &
  0.0204 &
  \underline{0.0262} &
  0.0146 &
  0.0164 & 0.0160

  \\
Ditto &
  \underline{0.0124} &
  \underline{0.0150} &
  \underline{0.0095} &
  0.0103 &
  \underline{0.0205} &
  0.0258 &
  \underline{0.0148} &
  \underline{0.0166} & \underline{0.0164}

 \\ \midrule
\model &
\textbf{0.0130} &
\textbf{0.0156} &
\textbf{0.0096} &
\textbf{0.0106} &
\textbf{0.0216} &
  \textbf{0.0272} &
  \textbf{0.0154} &
  \textbf{0.0169} & \textbf{0.0173} \\  \bottomrule
\end{tabular}
\end{adjustbox}}
\end{table*}
\begin{proposition}[Gradient descent down-weights harmful directions]\label{prop:downweight}
Fix client $i$ and directions $\{D_j\}_{j=1}^N$. Consider gradient descent on $\alpha_i$ with step size $\eta_\alpha>0$:
\begin{equation}
    \alpha_{ij}^{(t+1)}
    =
    \alpha_{ij}^{(t)}
    - \eta_\alpha
    \frac{\partial F_i(W_0 + \Delta W_i(\alpha_i^{(t)}))}{\partial \alpha_{ij}}.
\end{equation}
Assume that at some iterate $\alpha_i^{(t)}$ we have
\begin{equation}
    \mathrm{sign}\!\left(
        \big\langle
            \nabla F_i\big(W_0 + \Delta W_i(\alpha_i^{(t)})\big),
            \, D_j
        \big\rangle_F
    \right)
    =
    \mathrm{sign}(g_{ij}).
\end{equation}
Then:
\begin{itemize}
    \item If $g_{ij} > 0$ (locally, moving along $+D_j$ increases $F_i$ at $W_0$, i.e., direction $D_j$ is harmful), then 
    \begin{equation}
        \alpha_{ij}^{(t+1)} < \alpha_{ij}^{(t)}
    \end{equation}.
\item If $g_{ij} < 0$ (locally, moving along $+D_j$ decreases $F_i$ at $W_0$, i.e., direction $D_j$ is beneficial), then 
\begin{equation}
    \alpha_{ij}^{(t+1)} > \alpha_{ij}^{(t)}
\end{equation}

\end{itemize}
\end{proposition}

\begin{proof}
By Lemma~\ref{lem:grad-alpha},
\begin{equation*}
    \alpha_{ij}^{(t+1)}
    =
    \alpha_{ij}^{(t)}
    - \eta_\alpha
    \big\langle
        \nabla F_i\big(W_0 + \Delta W_i(\alpha_i^{(t)})\big),
        \, D_j
    \big\rangle_F.
\end{equation*}
If the inner product is positive (harmful direction), the update decreases $\alpha_{ij}$; if it is negative (beneficial direction), the update increases $\alpha_{ij}$. The sign assumption ensures consistency with $g_{ij}$ at $W_0$.
\end{proof}

Proposition~\ref{prop:downweight} shows that, under gradient-based optimization, personalized weights $\alpha_{ij}$ automatically suppress directions that increase the loss on client $i$ and emphasize those that decrease it. Combined with the higher personalization capacity in Proposition~\ref{prop:personalization}, this explains why the proposed decomposition of LoRA updates into shared directions and personalized weights is able to mitigate negative transfer across heterogeneous clients.

\begin{table}[h]
\centering
\caption{Statistics of the preprocessed GoodReads and Amazon datasets.}
\label{tab:dataset-stats}
\scalebox{0.80}{ 
\begin{tabular}{lcccc}
\toprule
\textbf{Domain} & \textbf{\#Users} & \textbf{\#Items} &  \textbf{\#Avg Inter.}\\
\midrule
Goodreads - Crime     & 5,000 & 31031 & 21.42 \\
Goodreads - Comics    & 5,000 & 15636  & 24.31\\
Goodreads - Children  & 5,000 & 21422 & 26.88 \\
Amazon - Beauty       & 4,000 & 10077   & 8.92\\
Amazon - Clothing     & 4,000 & 14087 & 9.35\\
Amazon - Electronics       & 5,000 & 8661   & 6.87\\
Amazon - Phones     & 5,000 & 23037 & 9.89\\
\bottomrule
\end{tabular}
}
\end{table}
\section{Experimental Setting}

\subsection{Datasets}
\label{appendix:data}
We introduce our datasets in detail here. As mentioned in the paper, we conduct experiments on the Goodreads and Amazon review datasets, covering three domains from Goodreads (\textit{Crime}, \textit{Comics}, \textit{Children}) and two domain pairs from Amazon (\textit{Beauty} \& \textit{Clothing}, \textit{Electronics} \& \textit{Phones}). For each user in the dataset, we use the last interaction for testing, the second-to-last for validation, and the rest for training. For each domain, we randomly sample 5,000 users from the Goodreads dataset and 4,000 users from the Amazon dataset.  Table~\ref{tab:dataset-stats} summarizes the statistics of each dataset after preprocessing.

\subsection{Efficient Local Storage Implementation}
\label{storage_efficiency}
Due to the large size of generative recommendation models, many local clients may not have sufficient storage or computational resources to host the full model. To address this, in our implementation, we adopt a two-step strategy to prepare a lightweight but effective local model.

First, we finetune the pre-trained base model \( \mathcal{M}_{\text{base}} \) on a small recommendation dataset \( \mathcal{D}_{\text{pretrain}} \) (the pre-training datasets mentioned in Implementation Details in Sec 4.1), resulting in a new model \( \mathcal{M}_{\text{ft}} \). This step helps the model become familiar with the input format and content structure of downstream recommendation tasks.

Second, inspired by prior distillation-based compression techniques~\cite{xuslmrec}, we compress \( \mathcal{M}_{\text{ft}} \) into a smaller local model \( \mathcal{M}_{\text{local}} \) by reducing its depth. Specifically, for every two layers \( f_{2i-1} \) and \( f_{2i} \) in \( \mathcal{M}_{\text{ft}} \), we construct a single layer \( g_i \) in \( \mathcal{M}_{\text{local}} \) to approximate their combined transformation:
\[
g_i(x) \approx f_{2i}(f_{2i-1}(x))
\]

To preserve model behavior during this compression, we apply \textit{intermediate-layer supervision}. Instead of matching only the final outputs, we minimize the discrepancy between the output of each compressed layer \( g_i \) and the corresponding two-layer composition \( f_{2i} \circ f_{2i-1} \) in the original model. The distillation objective is defined as:
\[
\mathcal{L}_{\text{distill}} = \sum_{(x, y) \in \mathcal{D}_{\text{pretrain}}} \sum_{i=1}^{L/2} \left\| g_i(h_i) - f_{2i}(f_{2i-1}(h_i)) \right\|_2^2
\]
where \( h_i \) denotes the input to the \((2i-1)\)-th layer in \( \mathcal{M}_{\text{ft}} \), and \( L \) is the total number of layers in \( \mathcal{M}_{\text{ft}} \).

This distillation process enables the compressed model to retain the expressive power of the original one while significantly reducing the local storage and computation overhead.

\section{Additional Experiment Results}
\label{appendix:add_results}
To further validate the effectiveness of our proposed method, we present the results on \textit{Amazon Electronics $\&$ Phones} and compare it with several strong baselines. As shown in Table~\ref{tab:main-add}, \model{} consistently outperforms the baselines, including FedAvg, PFedAvg, FedProx, and Ditto, across all evaluation metrics on both domains.

\begin{table}[h]
\centering
\caption{Comparison of average local computation cost and convergence efficiency across methods.}
\label{tab:computation_cost}
\scalebox{1}{
\begin{tabular}{lcc}
\toprule
\textbf{\footnotesize Method} &
\makecell{\footnotesize \textbf{Avg. Local}\\[-2pt]\footnotesize \textbf{Train Time / Round (s)}} &
\makecell{\footnotesize \textbf{Converge}\\[-2pt]\footnotesize \textbf{Round}}
 \\
\midrule
PFedAvg        & 688& 15 \\

Ditto           & 1665 & 7 \\

FDLoRA          & 1185 & 10 \\
FellRec            & 850 & 16 \\
\model~     & 746 & 12  \\
\bottomrule
\end{tabular}}
\end{table}

\section{Computation and Convergence Efficiency}

\label{appendix:computation}
Table~\ref{tab:computation_cost} reports the average local training time per round and the number of rounds to convergence for different frameworks. Compared with existing baselines, \model~achieves comparable convergence efficiency while maintaining a lower local training cost. In particular, Ditto requires dual-model optimization, and FDLoRA decouples local model learning from fusion-function learning, both of which increase local computation. 
This demonstrates that \model~strikes a favorable balance between personalization and computational efficiency, enabling scalable deployment under limited client resources.

\end{document}